\documentclass[acmsmall,screen,nonacm]{acmart}
\usepackage[T1]{fontenc}
\usepackage{graphicx}%
\usepackage{multirow}%
\usepackage{amsmath,amsfonts}%
\allowdisplaybreaks

\setcopyright{none}

\usepackage{mathrsfs}%
\usepackage[title]{appendix}%
\usepackage{xcolor}%
\usepackage{textcomp}%
\usepackage{manyfoot}%
\usepackage{booktabs}%
\usepackage{comment}
\usepackage{algorithm}%
\usepackage{algorithmicx}%
\usepackage{algpseudocode}%
\usepackage{listings}%
\usepackage{microtype}
\usepackage{wrapfig}
\usepackage{subcaption}
\usepackage{amsthm}
\usepackage{enumerate}

\usepackage{kotex}
\usepackage[colorinlistoftodos,prependcaption,obeyFinal]{todonotes}

\newtheorem{theorem}{Theorem}
\newtheorem{definition}{Definition}

\presetkeys{todonotes}{inline, color=topaz!40, bordercolor=topaz}{}

\newcommand{\sysname}{SAVED} 

\title{On the Decompositionality of Neural Networks}
\author{Junyong Lee}
\affiliation{\institution{Yonsei University}\country{South Korea}}
\email{rememberiom@yonsei.ac.kr}

\author{Baek-Ryun Seong}
\affiliation{\institution{University of Seoul}\country{South Korea}}
\email{sizzflair97@uos.ac.kr}

\author{Sang-Ki Ko}
\affiliation{\institution{University of Seoul}\country{South Korea}}
\email{sangkiko@uos.ac.kr}

\author{Andrew Ferraiuolo}
\affiliation{\institution{Independent Researcher}\country{UK}\authornote{Now at Google DeepMind}}
\email{andrew@andrewferr.org}

\author{Minwoo Kang}
\affiliation{\institution{Gwangju Institute of Science and Technology}\country{South Korea}}
\email{fullbodywin3@gm.gist.ac.kr}

\author{Hyuntae Jeon}
\affiliation{\institution{Independent Researcher}\country{South Korea}}
\email{duliuli@naver.com}

\author{Seungmin Lim}
\affiliation{\institution{Independent Researcher}\country{South Korea}}
\email{rsm1021902@gmail.com}

\author{Jieung Kim}
\affiliation{\institution{Yonsei University}\country{South Korea}}
\email{jieungkim@yonsei.ac.kr}
\begin{document}

\begin{abstract}
Recent advances in deep neural networks have achieved state-of-the-art performance across vision and natural language processing tasks. In practice, however, most models are treated as monolithic black-box functions, limiting maintainability, component-wise optimization, and systematic testing and verification. Although various empirical approaches have explored pruning and network decomposition, the field still lacks a principled semantic notion of when a neural network can be meaningfully decomposed.

In this work, we introduce a formal notion of neural network decompositionality defined as a semantic-preserving abstraction over neural architectures. Our key insight is that decompositionality should be characterized by the preservation of semantic behavior along the model’s decision boundary, which determines classification outcomes. This perspective provides a semantic contract between the original model and its decomposed components, enabling a rigorous formulation of decompositionality.

Building on this formulation, we develop a boundary-aware decomposition framework, \sysname\ (\textbf{S}emantic-\textbf{A}ware \textbf{Ve}rification-Driven \textbf{D}ecomposition), that instantiates the proposed required concepts in the decompositionality formal definition. The framework combines counterexample mining over low logit-margin inputs, probabilistic coverage of the input space, and structure-aware pruning to construct candidate decompositions that preserve decision-boundary semantics.

We evaluate the framework across multiple model families, including CNNs, language Transformers, and Vision Transformers. Our empirical study reveals clear differences in their main tasks. Language Transformers largely preserve semantic boundaries under decomposition, whereas vision models frequently violate the proposed decompositionality criterion, suggesting intrinsic barriers to decomposition in many visual tasks. These results establish decompositionality as a formally definable and empirically testable property of neural networks, providing a principled foundation for modular reasoning about neural architectures and exposing architecture-specific limits that remain invisible under purely structural notions of decomposition.

\end{abstract}

\maketitle              

\section{Introduction}\label{sec:intro}

Deep learning has become a fundamental technology across a wide range of computing systems, including autonomous systems, medical decision support, and large-scale language services~\cite{LeCunBH15, Goodfellow-et-al-2016, otter20, Shamshad23, GrigorescuTCM20}. Despite their widespread deployment, neural networks remain difficult to engineer using established software engineering principles. In conventional software systems, modularity enables local reasoning. Systems can be decomposed into components that can be analyzed, modified, and reused independently while preserving global correctness. Neural networks, in contrast, are typically treated as monolithic functions. Although architectural constructs such as layers, channels, and attention heads introduce structural segmentation, it remains unclear whether these boundaries correspond to semantically meaningful units. As a result, many software engineering techniques including local verification, component-level repair, modular reuse, and systematic optimization do not transfer naturally to neural-network–based systems.

\paragraph{\textbf{The missing formal notion of neural decompositionality.}}
A growing body of work explores decomposing trained networks into modules, motivated by goals such as reuse, model compression, separation of concerns, or accelerating downstream analyses. Empirical results suggest that multi-class classifiers can sometimes be partitioned into class-wise or group-wise submodels while preserving test accuracy mainly on benchmark datasets~\cite{PanR20, PanR22, ImtiazBSPCR23, Ren23}. These efforts highlight the practical potential of neural network decomposition and raise important questions about its validity. However, they primarily investigate the phenomenon empirically and conceptually, rather than establishing a precise formal model of when such decompositions are semantically valid. Their evaluations rely primarily on test accuracy as the indicator of decomposition quality. While useful in practice, this introduces a heuristic assumption: that preserving test accuracy implies preservation of the model’s semantic behavior. However, test accuracy provides only a sparse approximation of the input domain and therefore cannot guarantee such preservation as described in Figure~\ref{fig:emperical-decomposition-and-true-decomposition}.  This limitation becomes particularly evident near decision boundaries, where the classifier’s semantic transitions occur. In particular, two models may achieve similar aggregate accuracy while exhibiting substantial differences in the geometry of their decision boundaries—the regions where class semantics transition. Consequently, existing work largely focuses on \emph{how} to decompose networks, while the more fundamental question remains open:
\textbf{\emph{``When does a neural network admit a semantically meaningful decomposition?''}}

\begin{wrapfigure}{r}{0.5\textwidth}
    \centering
    \includegraphics[width=0.5\textwidth]{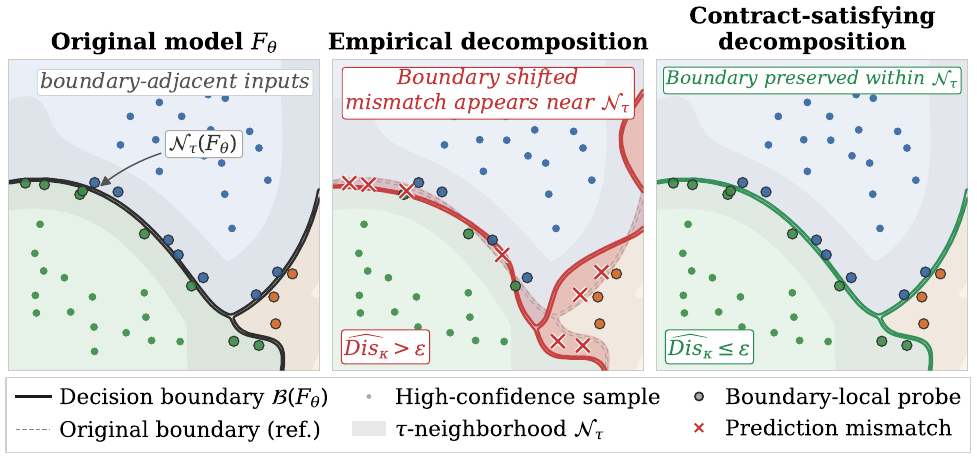}
    \caption{\textbf{Empirical vs. true decomposition.} Empirical decomposition can maintain aggregate performance while introducing behavioral inconsistencies near decision boundaries. True decomposition instead preserves boundary-local semantics, ensuring that decomposed components maintain the original model’s classification behavior around semantic transitions.}
    \label{fig:emperical-decomposition-and-true-decomposition}
\end{wrapfigure}

\paragraph{\textbf{Our position: Decompositionality is a property, not a procedure.}}
This work answers this question by reframing decomposition as an intrinsic, \emph{semantically grounded} property of a trained model. In this manner, we introduce \emph{neural decompositionality}, a formally specified and empirically testable contract that characterizes when a model can be decomposed into structurally distinct components while preserving semantic behavior. Our starting point is that for multi-class classifiers, semantics are encoded by the partition of the input space and, most critically, by the \emph{decision boundary}~\cite{KarimiTD19,KarimiD22,FawziMFS17,Oyallon17} where class transitions occur. Requiring global functional equivalence between the original and decomposed models is unnecessarily strong (and typically infeasible to establish), while relying on aggregate accuracy is too weak. We therefore target \emph{boundary-local semantic stability}: whether structured interventions on model components preserve class-wise behavior in neighborhoods of semantic transitions.

\paragraph{\textbf{Semantic boundary and boundary-local contracts.}}
We formalize decompositionality for multi-class classification networks through a boundary-based notion of semantic preservation. Concretely, we define the classifier’s decision boundary and its $\tau$-neighborhood $\mathcal N_\tau(F_\theta)$. Decompositionality is then specified as a \emph{semantic-structural contract} between a reference classifier $F_\theta$ and a decomposition $F_\theta^\downarrow$.
The contract imposes two orthogonal requirements.
First, \emph{boundary-aware semantic fidelity} requires that the aggregated decomposed predictor agrees with the reference model on $\mathcal N_\tau(F_\theta)$ up to $\varepsilon$ disagreement.
Second, \emph{structural divergence} requires that distinct components exhibit non-trivial representational divergence over the same boundary-relevant region: pairwise component overlap must remain below a threshold $\gamma$, and each component must be non-trivially smaller than the original model (controlled by $\eta$). This dual requirement is intentional. Semantic fidelity alone admits degenerate identity-like decompositions, whereas structural differentiation alone permits semantically incorrect partitions. Together, these conditions characterize decompositions that are both behaviorally meaningful and structurally non-redundant.  This contract therefore captures when a neural network admits a decomposition that preserves its semantic decision structure.

\begin{figure*}
    \centering
    \includegraphics[width=1.0\textwidth]{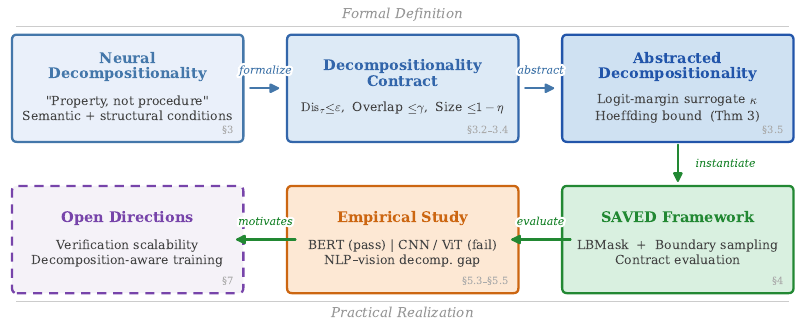}
    \caption{\textbf{Overview of neural decompositionality.}
        We view decompositionality as a semantic property of a trained model rather than a purely procedural transformation.
        Starting from a boundary-based semantic contract that characterizes when decomposition preserves the classifier’s decision structure,
        we derive an abstracted formulation that enables tractable reasoning over boundary-relevant regions.
        This abstraction is instantiated through \sysname, a boundary-aware decomposition framework that probes candidate decompositions using counterexample-guided boundary sampling.
        Finally, we evaluate the resulting decompositions across multiple architectures and demonstrate that when decompositionality holds, the resulting components improve the scalability of downstream neural network analysis tasks such as verification.}
    \label{fig:overview}
\end{figure*}

\paragraph{\textbf{From global property to analyzable abstraction.}}
Figure~\ref{fig:overview} illustrates the overall structure of our approach.
Starting from a boundary-based semantic contract for neural decompositionality,
we derive an abstracted formulation that enables tractable reasoning about boundary-local behavior.
This abstraction provides the theoretical foundation for \sysname,
a framework that operationalizes the proposed contract through boundary-aware probing and counterexample discovery.

While the preceding contract provides a precise semantic characterization of decompositionality, directly verifying such a property is challenging in practice. Decision boundaries and input domains are continuous and effectively unbounded, making it computationally infeasible to verify a universal decompositionality claim in its concrete form. To bridge the gap between semantic idealization and practical evaluation, we proceed in two stages. We first specify an idealized \emph{global} decompositionality contract defined directly over boundary neighborhoods under an input distribution. We then introduce \emph{abstracted decompositionality}, a soundly parameterized relaxation that preserves the semantic intent of the global definition while enabling tractable reasoning. This abstraction also remains a boundary-based concept. It concentrates analysis on boundary-relevant regions and supports executable procedures for estimating or falsifying decompositionality claims.

\paragraph{\textbf{From formal definition to practical realization.}}
A definition is most valuable when it admits a concrete theoretical foundation that can be realized in practice. Purely semantic characterizations clarify conceptual properties, but without a pathway toward realization, they remain difficult to apply to practical problems where neural network modularization could provide feasible benefits. To evaluate whether a candidate decomposition satisfies the boundary-local contract, we develop a boundary-aware counterexample probing strategy that prioritizes inputs with small logit margins while maintaining probabilistic coverage over the input space. Low-margin regions are where semantic violations are most likely to occur, yet coverage is necessary to avoid overfitting the probe to a narrow slice of the domain.

\paragraph{\textbf{What we learn across architectures and tasks.}}
To demonstrate our framework and validate the concept, we apply it to CNNs, NLP Transformers, and Vision Transformers (ViTs) across multiple classification tasks, examining when and where neural decompositionality emerges in practice. Our primary objective is to empirically assess whether boundary-local semantic preservation holds consistently across architectures. Our results reveal a clear and consistent pattern. Boundary-level semantic preservation is systematically achieved in language tasks, whereas vision tasks exhibit systematic semantic violations near decision boundaries. To better understand these discrepancies, we further analyze the behavior under diverse decomposition configurations. Importantly, when decompositionality holds, the resulting decomposed models yield tangible benefits for downstream analysis. In particular, decomposition significantly improves the scalability of neural network analysis in verification-oriented settings, as reasoning over smaller semantic components enables more efficient verification procedures compared to operating on the original monolithic model.

All in all, this paper makes the following contributions:
\begin{itemize}

    \item \textbf{A formal and tractable notion of neural decompositionality.}
          We introduce \emph{neural decompositionality}, a formally defined property that characterizes when a neural network admits a semantically meaningful decomposition. Our formulation models decompositionality as a semantic--structural contract associated with the classifier’s decision boundary. Under this view, decompositionality is instantiated through two complementary conditions, boundary-aware semantic fidelity and structural divergence among components. To enable practical reasoning, we further introduce \emph{abstracted decompositionality}, a boundary-centric relaxation that preserves the semantic intent of the original contract while allowing executable evaluation in continuous input spaces. To the best of our knowledge, this is the first work to formalize neural network decomposition through decision-boundary–centric semantic definitions.

    \item \textbf{A boundary-aware evaluation methodology.}
          We build a framework that faithfully reflects the proposed concept. The framework performs boundary-aware counterexample probing and learning-and-mask-based decomposition.

    \item \textbf{Empirical insights across architectures and tasks.}
          Through experiments on CNNs, NLP Transformers, and ViT across multiple classification tasks, we analyze when neural decompositionality emerges in practice. Our results reveal systematic differences across domains and show that, when decompositionality holds, decomposition leads to improved scalability for downstream neural network analysis.

\end{itemize}

The remainder of this paper is organized as follows.
Section~\ref{sec:background} presents the background and motivation for neural network decomposition and semantic reasoning about neural models.
Section~\ref{sec:decomposition} introduces the formal definition of neural decompositionality, first at the global level and then through its abstracted formulation.
Section~\ref{sec:saved} presents \sysname\ (\textbf{S}emantic-\textbf{A}ware \textbf{Ve}rification-Driven \textbf{D}ecomposition), our framework that operationalizes the proposed decompositionality verifier.
Section~\ref{sec:evaluation} reports the experimental evaluation and analysis using \sysname.
Section~\ref{sec:related-works} reviews related work, and Section~\ref{sec:conclusion} concludes the paper.
All core artifacts associated with this work are publicly available at \url{https://zenodo.org/records/19049545}.

\section{Background}
\label{sec:background}

This section introduces the basic concepts used throughout the paper.
We first describe the \textbf{functional formulation of neural network classifiers}.
We then discuss the \textbf{distinction between structural and semantic decomposition},
highlighting the limitations of existing structural approaches.
Finally, we explain how the \textbf{decision boundary} captures the semantic behavior of a classifier, motivating the boundary-based definition of neural decompositionality introduced in the next section.

\subsection{Neural Network Classifiers}

We model a neural network classifier as a function, following the conventions presented in multiple prior works~\cite{HuangKRSSTWY20, albarghouthi2021introduction, Goodfellow-et-al-2016, LeCunBH15}.
\[
    F_\theta : \mathcal{X} \rightarrow \mathbb{R}^{C},
\]
where \(\theta\) denotes the set of parameters of the network. For an input \(x \in \mathcal{X}\), the classifier produces a vector of logits
\[
    z_\theta(x) := F_\theta(x),
\]
where \(z_{\theta,i}(x)\) represents the logit associated with class \(i\). The predicted label is given by
\[
    \hat{y}_\theta(x) = \underset{i \in \{1,\dots,C\}}{\arg\!\max}\ z_{\theta,i}(x).
\]
which maps an input \(x \in \mathcal{X}\) to a vector of logits over \(C\) classes. Neural networks implement this mapping through a sequence of intermediate transformations.
For a network with \(L\) layers, we denote the representation at layer \(\ell\) by
\[
    h_\ell(x) \in \mathbb{R}^{d_\ell},
\]
where \(d_\ell\) denotes the dimension of the representation space at layer \(\ell\).
The network computation can therefore be expressed recursively as
\[
    h_\ell(x) = \phi_\ell(h_{\ell-1}(x)),
    \qquad h_0(x) = x,
\]
where each transformation \(\phi_\ell\) represents the operation performed at layer \(\ell\).
The final logits are produced from the last representation by a classifier head
\[
    z_\theta(x) = g(h_L(x)).
\]

This formulation highlights that, regardless of the architectural mechanisms used to implement them (e.g., convolution, attention~\cite{VaswaniSUAJPGP17}, or residual connections~\cite{He16}), neural networks ultimately define a function that maps inputs in \(\mathcal{X}\) to class scores in \(\mathbb{R}^C\).
From this perspective, reasoning about the behavior of a neural network amounts to understanding how its internal transformations collectively determine the resulting classification decision.

\subsection{Structural vs. Semantic Decomposition}


Many existing techniques implicitly attempt to decompose neural networks into smaller components. Examples include pruning methods that remove redundant parameters~\cite{Cheng24, Sun0BK24}, modular architectures such as mixture-of-experts models~\cite{JacobsJNH91, shazeer2017}, and task-specific specialization of subnetworks~\cite{frankle2019, PanR20, PanR22}. These approaches can be viewed as forms of \emph{structural decomposition}, where a model is partitioned or simplified according to architectural criteria. However, structural decomposition provides limited guarantees about how the functional behavior of the model changes. For instance, pruning or model compression techniques may alter the classifier's behavior in previously unseen regions of the input space, even if predictive accuracy is preserved on a validation dataset. As a result, structural transformations do not necessarily correspond to well-defined semantic relationships between the original and modified models.

This limitation contrasts with classical software systems. They rely heavily on \emph{semantic modularity}. Components interact through well-defined interfaces and behavioral contracts, enabling reasoning about system correctness through local analysis of individual modules. Such modular reasoning is largely absent in current neural-network systems, where decomposition techniques are typically guided by structural heuristics rather than formally defined semantic properties. These observations motivate the need for a semantic notion of decomposition for neural networks. Rather than focusing solely on architectural structure, a semantic decomposition should characterize when a model can be partitioned into components while preserving its functional behavior. The next section introduces a formal definition of decompositionality that captures this notion.

\subsection{Decision Boundaries and Semantic Behavior}

For classification models, semantic behavior is determined by how the classifier partitions the input space into regions associated with different
labels. These regions are separated by the model's \emph{decision boundaries}~\cite{FawziMFS17, KarimiD22, KarimiTD19, Oyallon17, DBLP:conf/iclr/XuSGGH23, ZhaoNG24}. They represent locations where small perturbations to the input may change the predicted class. Consequently, reasoning about semantic preservation under model transformations requires understanding how the decision boundary changes. Two models may achieve similar aggregate accuracy while exhibiting substantial differences in the geometry of their decision boundaries. Such differences may lead to inconsistent behavior in regions of the input space where class transitions occur. This observation motivates using decision boundaries as the reference structure for defining semantic preservation. In the next section, we introduce a formal definition of neural decompositionality that characterizes when a decomposition preserves the decision behavior of a classifier while producing structurally distinct components.

\section{Defining Decompositionality}
\label{sec:decomposition}

Building on the neural network classifier formulation introduced in Section~\ref{sec:background}, we now formalize the notion of \emph{decompositionality}.

\subsection{Decision Boundary}

Let \(F_\theta : \mathcal X \to \mathbb R^C\) be a classifier as defined
in Section~\ref{sec:background}, and let
\(\hat y_\theta(x)\) denote its predicted label.
The decision boundary characterizes the set of inputs at which the classifier becomes unstable with respect to class prediction.

\subsubsection{Pairwise Decision Boundary}

For each pair of distinct classes \(i \neq j\), the
\emph{pairwise decision boundary} between classes \(i\) and \(j\) is
defined as
\[
\mathcal B_{i,j}(F_\theta)
=
\left\{
x \in \mathcal X
\;\middle|\;
\forall \varepsilon > 0,\;
\exists x' \in \mathcal X
\text{ such that }
\|x-x'\|<\varepsilon
\text{ and }
\{\hat y_\theta(x),\hat y_\theta(x')\}=\{i,j\}
\right\}.
\]
Intuitively, \(\mathcal B_{i,j}(F_\theta)\) contains points at which an arbitrarily small perturbation can switch the prediction between classes \(i\) and \(j\). This directly follows the meaning of decision boundaries described in Section~\ref{sec:background}.

\subsubsection{Global Decision Boundary}

The decision boundary of a multiclass classifier can therefore be
constructed by combining the boundaries between all class pairs.
The global decision boundary of the classifier is the union of all
pairwise decision boundaries:
\[
\mathcal B(F_\theta)
=
\bigcup_{i \neq j}
\mathcal B_{i,j}(F_\theta).
\]
Thus, \(\mathcal B(F_\theta)\) contains all input points at which an
arbitrarily small perturbation may change the predicted label.

\subsection{Boundary-Aware Semantic Fidelity}
We want to characterize how well a decomposition performs relative to the original.
To do so, we must understand how often its decisions are changed relative to the original.
 The decision boundary identifies regions where semantic transitions occur. Away from the boundary, classifier predictions tend to be locally stable, meaning that small perturbations do not change the predicted class. Consequently, evaluating semantic preservation over the entire input space may obscure the regions where the classifier's behavior is most sensitive. To capture this phenomenon, we evaluate semantic fidelity in neighborhoods around the decision boundary of the original model.

\subsubsection{Boundary Neighborhood}

Let \(\mathcal B(F_\theta) \subseteq \mathcal X\) denote the decision boundary induced by the classifier \(F_\theta\). For a threshold \(\tau > 0\), we define the \(\tau\)-neighborhood of the boundary as
\[
\mathcal N_\tau(F_\theta)
=
\{x \in \mathcal X \mid \mathrm{dist}(x,\mathcal B(F_\theta)) \le \tau\},
\]
where
\[
\mathrm{dist}(x,S) = \inf_{z \in S} \rho(x,z)
\]
denotes the distance between a point \(x\) and a set \(S\) under a metric
\(\rho\), such as an \(\ell_p\) metric on \(\mathbb{R}^{d_x}\).




\subsubsection{Boundary-Aware Disagreement}

Let \(F_\theta\) denote the original classifier introduced in Section~\ref{sec:background}. We consider a decomposition of \(F_\theta\) into a collection of component predictors. Formally, let
\[
F_\theta^\downarrow(x)
=
A\big(m_1(x),\dots,m_K(x)\big)
\]
denote a decomposed model derived from \(F_\theta\), where each component
\(
m_k : \mathcal X \rightarrow \mathbb R
\)
produces a scalar score (e.g., a logit) associated with a designated positive class set \(P_k \subseteq [C]\). Intuitively, \(m_k(x)\) measures the affinity of the input \(x\) to the class subset \(P_k\). The aggregation operator \(A\) combines the component scores \(\{m_k(x)\}_{k=1}^K\) to produce the final multiclass prediction. Let \(\hat y_\theta(x)\) and \(\hat y_\theta^\downarrow(x)\) denote the predictions induced by the original classifier \(F_\theta\) and the decomposed model \(F_\theta^\downarrow\), respectively. Let \(\mathcal D\) be an input distribution over \(\mathcal X\). We define the boundary-aware conditional disagreement probability between the two models as
\[
\mathrm{Dis}_\tau\left(F_\theta^\downarrow \middle| F_\theta\right)
=
\Pr_{x \sim \mathcal D}
\left[
\hat y_\theta^\downarrow(x)
\neq
\hat y_\theta(x)
~\middle\vert~
x \in \mathcal N_\tau(F_\theta)
\right].
\]
This quantity measures the probability that the decomposed model
produces a different prediction from the original classifier
within the \(\tau\)-neighborhood of the decision boundary.
Intuitively, it quantifies how often the decomposition alters the
classifier's behavior in regions where semantic transitions occur.

\subsubsection{Boundary-Aware Semantic Fidelity}
We finally define semantic preservation under decomposition. This condition ensures that the decomposed model preserves the
classification behavior of the original model with high probability
within regions close to the decision boundary, where label transitions
are most sensitive.

\begin{definition}\label{def:boundary-aware-semantic-fidelity}
Let \(F_\theta\) be a classifier and \(F_\theta^\downarrow\) a decomposed model.
We say that \(F_\theta^\downarrow\) satisfies
\emph{\((\varepsilon,\tau)\)-boundary-aware semantic fidelity}
with respect to \(F_\theta\) if
\[
\mathrm{Dis}_\tau\left(F_\theta^\downarrow \middle| F_\theta\right)
\le
\varepsilon .
\]   
\end{definition}

\subsubsection{Empirical Estimation}
In practice, the disagreement probability can be estimated empirically
using samples drawn from the input distribution \(\mathcal D\).
Given a dataset \(S \sim \mathcal D\), one can approximate
\(\mathrm{Dis}_\tau(F_\theta^\downarrow \mid F_\theta)\) by measuring
prediction disagreements restricted to samples that lie within the
boundary neighborhood \(\mathcal N_\tau(F_\theta)\).
Standard evaluation metrics such as accuracy or class-wise F1-score
computed over this restricted set provide practical estimates of
boundary-aware semantic fidelity.

\subsubsection{Boundary Preservation Theorem}

The definition of boundary-aware semantic fidelity leads to the
following property.
If a decomposed model agrees with the original classifier in a
sufficiently small neighborhood of the decision boundary, then the two
models induce nearly identical classification behavior globally.

\begin{theorem}[Boundary Preservation]\label{thm:boundary-preservation}
Let \(F_\theta\) be a classifier and \(F_\theta^\downarrow\) a decomposed
model satisfying \((\varepsilon,\tau)\)-boundary-aware semantic fidelity
with respect to \(F_\theta\).
Suppose that the input distribution \(\mathcal D\) assigns at most
\(\delta\) probability mass to regions outside the boundary neighborhood
\(\mathcal N_\tau(F_\theta)\) where label disagreement may occur.
Then the overall disagreement probability satisfies
\[
\Pr_{x \sim \mathcal D}
\left[
\hat y_\theta^\downarrow(x)
\neq
\hat y_\theta(x)
\right]
\le
\varepsilon + \delta .
\]
\end{theorem}
\begin{proof}
The total disagreement probability can be decomposed into two regions of
the input space: the boundary neighborhood \(\mathcal N_\tau(F_\theta)\)
and its complement.
\begin{itemize}
    \item Within the boundary neighborhood, disagreement is bounded by
\(\varepsilon\) by the definition of boundary-aware semantic fidelity.
\item Outside this region, disagreements can occur only with probability
bounded by the distributional mass \(\delta\).
\end{itemize}
Combining the bounds from the two regions yields the stated result.
\end{proof}

\subsection{Structural Divergence}
Semantic fidelity alone does not guarantee that a decomposition produces meaningful modular structure. A trivial construction could replicate the original network multiple times or produce components that reuse nearly identical internal representations. Such constructions preserve predictions but do not yield genuine modular decomposition. To capture the other part of decomposition, we introduce the notion of \emph{structural divergence}. 
This notion is agnostic to the specific decomposition mechanism and can be instantiated over parameters, neurons, or activation regions.
Intuitively, a valid decomposition should partition the internal computation of the model into components that rely on sufficiently distinct subsets of parameters or activations. Let \(M_k \subseteq \{1,\dots,N\}\) denote the set of active units (e.g., neurons, channels, or parameters) used by component \(m_k\), where \(N\) denotes the total number of units in the original network.

\subsubsection{Structural Disjointness.}

We require that different components operate on sufficiently distinct subsets of the model. 
Let \(S_i\) and \(S_j\) denote the structural support of components \(m_i\) and \(m_j\), respectively (e.g., parameters, neurons, or activation regions). 
We define the overlap ratio as
\[
\mathrm{Overlap}(S_i, S_j)
=
\frac{|S_i \cap S_j|}{|S_i \cup S_j|}.
\]
A decomposition satisfies structural disjointness if the following property holds for the threshold \(\gamma\):
\[
\mathrm{Overlap}(S_i, S_j) \le \gamma
\quad \forall i \neq j.
\]

\subsubsection{Non-Trivial Reduction.}
Each component must also represent a non-trivial reduction of the original model. Let
\[
\mathrm{Size}(S_k)
=
\frac{|S_k|}{|S|}
\]
denote the relative size of component \(m_k\), where \(S\) is the structural support of the original model. We require
\[
\mathrm{Size}(S_k) \le 1 - \eta
\]
for some threshold \(\eta > 0\).  Together with structural disjointness and non-trivial reduction, these conditions ensure that the resulting components are structurally distinct and non-trivially smaller than the original model. This formulation is agnostic to the specific decomposition mechanism and applies to a broad class of structural partitions.

\subsubsection{Empirical Estimation via Learned Masks.}

In practice, the structural support \(S_k\) of each component is not directly observable and must be approximated. 
We instantiate structural support using learned binary masks over model units, which provide a concrete and tractable representation of component-wise structure. Concretely, for each component \(m_k\), we associate a mask
\( M_k \in \{0,1\}^N \)
over the units of the original model, where \(M_k(i) = 1\) indicates that unit \(i\) is utilized by component \(m_k\). 
The structural support \(S_k\) is then defined as
\(
S_k := \{ i \mid M_k(i) = 1 \}.
\)
The masks \(\{M_k\}\) can be obtained through standard mask-learning or sparsification techniques, such as magnitude-based pruning~\cite{Han16}, gradient-based gating, or learned binary masking with straight-through estimators~\cite{Bengio13}. 
In practice, we jointly optimize the masks with respect to (i) task performance and (ii) sparsity or separation regularizers that encourage disjointness across components.












\subsection{Global Decompositionality}

We combine the two aforementioned concepts into a unified definition of decompositionality.

\begin{definition}\label{def:global-decomposition}
Let \(F_\theta\) be a neural network classifier and let
\[ F_\theta^\downarrow(x) = A(m_1(x), \dots, m_K(x)) \]
be a decomposition of \(F_\theta\) into components \(\{m_k\}_{k=1}^K\). We say that \(F_\theta\) satisfies \emph{$(\varepsilon,\tau,\gamma,\eta)$-global decompositionality} if the following conditions hold:
\begin{enumerate}
\item \textbf{Boundary-Aware Semantic Fidelity.}
The decomposed model \(F_\theta^\downarrow\) satisfies 
\[
\mathrm{Dis}_\tau\left(F_\theta^\downarrow \middle| F_\theta\right)
\le
\varepsilon.
\]
\item \textbf{Structural Disjointness and Reduction.}  
Let \(S_i\) denote the structural support of component \(m_i\), and let \(S\) denote that of the original model. The components satisfy
\[
\mathrm{Overlap}(S_i, S_j) \le \gamma
\quad \forall i \neq j\ \ \ \  \text{and} \ \ \ \ 
\frac{|S_k|}{|S|} \le 1 - \eta.
\]

\end{enumerate}
\end{definition}

This definition captures decompositionality as a joint property of \emph{semantic preservation} and \emph{structural separation}. The first condition ensures that the decomposition preserves classification behavior near decision boundaries, while the second enforces that components are both distinct and non-trivially reduced, ruling out degenerate identity-like decompositions.

\paragraph{Main Theorem.}

We show that global decompositionality yields both semantic and structural guarantees. Semantically, it preserves the behavior of the original classifier near its decision boundary and stabilizes standard evaluation metrics on boundary-restricted inputs. Structurally, it rules out degenerate decompositions in which components collapse to identical or near-complete copies of the original model.

\begin{theorem}[Semantic Preservation and Structural Non-Collapse]
\label{thm:semantic-preservation-and-structural-non-collapse}

Let \(F_\theta\) be a classifier and 
\(
F_\theta^\downarrow(x)=A(m_1(x),\dots,m_K(x))
\)
be a decomposition satisfying \((\varepsilon,\tau,\gamma,\eta)\)-global decompositionality. Then, the following properties hold:

\begin{enumerate}

\item \textbf{Boundary accuracy.}
\[
\Pr_{x\sim\mathcal D}
\left[
\hat y_{F_\theta^\downarrow}(x)=\hat y_{F_\theta}(x)
~\middle\vert~
x\in\mathcal N_\tau(F_\theta)
\right]
\ge 1-\varepsilon.
\]

\item \textbf{Boundary metric stability.}
Let \(C_\tau^\downarrow\), \(C_\tau\) denote the joint distributions of predicted and reference labels restricted to \(x\in\mathcal N_\tau(F_\theta)\). Then
\[
\left\|C_\tau^\downarrow - C_\tau\right\|_1 \le 2\varepsilon,
\]
since each disagreement affects at most two entries of the confusion matrix.

\item \textbf{Non-trivial reduction.}
\[
|S_k| \le (1-\eta)|S| \quad \forall k.
\]

\item \textbf{No structural collapse.}
\[
\mathrm{Overlap}(S_i,S_j) \le \gamma < 1 \quad \forall i\neq j.
\]

\item \textbf{Global disagreement.}
If off-boundary disagreement has mass at most \(\delta\), then
\[
\Pr_{x\sim\mathcal D}
\left[
\hat y_{F_\theta^\downarrow}(x)\neq \hat y_{F_\theta}(x)
\right]
\le \varepsilon+\delta.
\]

\end{enumerate}
\end{theorem}

\begin{proof}
Items 1 and 2 follow from boundary-aware semantic fidelity, which bounds disagreement by \(\varepsilon\). 
Items 3 and 4 follow from structural divergence and reduction. 
Item 5 follows by partitioning the input space and composing disagreement.
\end{proof}

\subsection{From Global to Local Decompositionality}

The global notion of decompositionality relies on the boundary neighborhood
\(\mathcal N_\tau(F_\theta)\) defined over the entire input space \(\mathcal X\). However, this is generally intractable to evaluate due to two challenges. First, the decision boundary \(\mathcal B(F_\theta)\) is implicitly defined in a high-dimensional space. Second, semantic fidelity requires reasoning over all inputs within \(\mathcal N_\tau(F_\theta)\).  

To obtain a tractable formulation, we are inspired by two complementary lines of work. From \emph{local robustness}, we adopt the principle of restricting analysis to neighborhoods around decision boundaries, where semantic transitions occur and verification is most critical. From \emph{abstract interpretation}, we borrow the idea of replacing intractable global reasoning with a finite, tractable abstraction. Concretely, we introduce a dataset-level abstraction that approximates the boundary neighborhood using a finite sample. Unlike classical abstract interpretation, our goal is not to construct a sound over-approximation, but to provide a boundary-focused surrogate that captures semantically relevant behavior in practice.

\subsubsection{Dataset-Induced Boundary Abstraction}
Let \(S=\{x_1,\dots,x_n\}\subset\mathcal X\) be sampled from \(\mathcal D\). 
We define the boundary subset
\[
\widetilde{\mathcal B}_\tau(F_\theta,S)
=
\{x\in S \mid \mathrm{dist}(x,\mathcal B(F_\theta))\le\tau\},
\]
where 
\[\mathrm{dist}(x,\mathcal B(F_\theta)) := \inf_{x'\in \mathcal B(F_\theta)} \|x-x'\|.\] 
Since this distance is not directly computable, we approximate it using logit margins.

\subsubsection{Logit-Margin Approximation}

Let \(y=\hat y_{F_\theta}(x)\). Define
\[
m(x)
=
F_{\theta,y}(x)
-
\max_{j\neq y} F_{\theta,j}(x).
\]
With this, inputs with small margins lie near the decision boundary. 
We define the margin-induced boundary subset
\[
\widehat{\mathcal B}_\kappa(F_\theta,S)
=
\{x\in S \mid m(x)\le\kappa\}.
\]

\subsubsection{Local Semantic Fidelity}
\label{sec:local-semantic-fidelity}

Using the margin-induced boundary subset, we define the empirical disagreement as
\[
\widehat{\mathrm{Dis}}_\kappa\left(F_\theta^\downarrow \middle| F_\theta; S\right)
=
\frac{1}{|\widehat{\mathcal B}_\kappa(F_\theta,S)|}
\sum_{x\in\widehat{\mathcal B}_\kappa(F_\theta,S)}
\mathbf{1}
\left[
\hat y_{F_\theta^\downarrow}(x)
\neq
\hat y_{F_\theta}(x)
\right],
\]
where \(\mathbf{1}[\cdot]\) denotes the indicator function that evaluates to 1 if the condition holds and 0 otherwise. 
This quantity estimates the disagreement between the decomposed and reference models over inputs near the decision boundary.

\begin{definition}[Local Semantic Fidelity]\label{def:local-semantic-fidelity}
A decomposition satisfies $(\varepsilon,\kappa)$-local semantic fidelity on $S$ if
\[
  \widehat{\mathrm{Dis}}_\kappa\!\bigl(F^\downarrow_\theta \mid F_\theta;\, S\bigr) \;\leq\; \varepsilon.
\]
\end{definition}

\subsubsection{Local Decompositionality}

\begin{definition}[Local Decompositionality]\label{def:local-decompositionality}
A decomposition satisfies $(\varepsilon,\kappa,\gamma,\eta)$-local decompositionality if
\begin{enumerate}
  \item $(\varepsilon,\kappa)$-local semantic fidelity holds (Definition~\ref{def:local-semantic-fidelity}), and
  \item structural divergence holds with parameters $(\gamma,\eta)$.
\end{enumerate}
\end{definition}
This provides a practical instantiation of decompositionality on finite data.

\subsubsection{Connection to Global Decompositionality}

\begin{theorem}[Local-to-Global Fidelity Approximation]\label{thm:local-global}
Let \(S\) be sampled i.i.d.\ from \(\mathcal D\). Then with probability at least \(1-\delta\),
\[
\left|
\mathrm{Dis}_\tau\left(F_\theta^\downarrow \middle| F_\theta\right)
-
\widehat{\mathrm{Dis}}_\kappa\left(F_\theta^\downarrow \middle| F_\theta; S\right)
\right|
\le
\sqrt{\frac{\log(2/\delta)}{2|S|}} .
\]

Consequently, if
\(
\ \widehat{\mathrm{Dis}}_\kappa\left(F_\theta^\downarrow \middle| F_\theta; S\right)
\le \varepsilon,
\)
then
\[
\mathrm{Dis}_\tau\left(F_\theta^\downarrow \middle| F_\theta\right)
\le
\varepsilon
+
\sqrt{\frac{\log(2/\delta)}{2|S|}} .
\]

\end{theorem}

\begin{proof}
Let
\[
X(x)=\mathbf{1}\left[\hat y_{F_\theta^\downarrow}(x)\neq \hat y_{F_\theta}(x)\right].
\]
Then \(\mathrm{Dis}_\tau = \mathbb{E}[X]\), while \(\widehat{\mathrm{Dis}}_\kappa\) is the empirical mean over \(S\). Since \(X\in[0,1]\), the result follows directly from Hoeffding's inequality.
\end{proof}

\section{SAVED: Decomposition Framework}\label{sec:saved}

\begin{figure}
    \centering
    \includegraphics[width=1\textwidth]{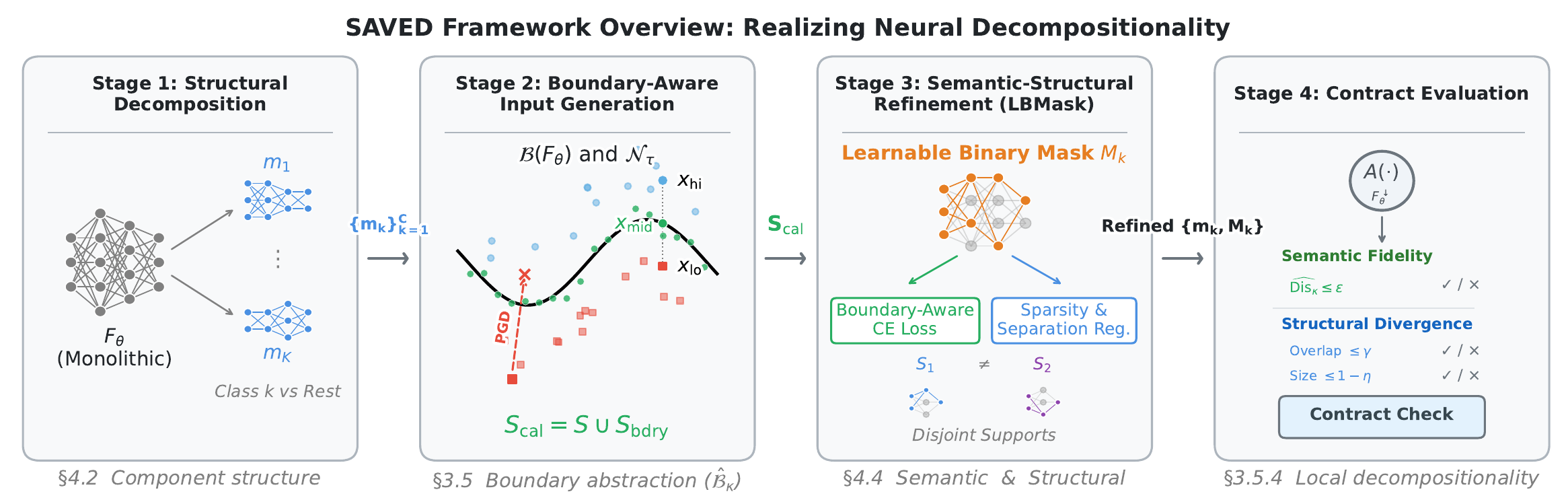}
\caption{\textbf{Overview of \sysname.}}
\label{fig:framework-overview}
\end{figure}

The formal definitions in Section~\ref{sec:decomposition} specify decompositionality as a semantic--structural property over the full input space. 
However, these definitions are not directly computable in practice. 
\sysname\ provides a \emph{practical realization} of this property on finite data by (i) constructing component-wise representations, (ii) approximating boundary-relevant inputs, (iii) learning masks that enforce semantic fidelity and structural separation, and (iv) evaluating the resulting decomposition against the local contract. 
This pipeline operationalizes decompositionality as a property that must be realized and verified on data, rather than assumed to hold by construction.

\subsection{Overview}

\sysname\ consists of four phases:
(1) \emph{structural decomposition},
(2) \emph{boundary-aware input generation},
(3) \emph{semantic-structural refinement (mask learning)}, and
(4) \emph{contract evaluation}.
These phases correspond directly to the abstract definition. 
Phase~1 constructs the component structure, Phase~2 approximates the boundary neighborhood, Phase~3 enforces semantic fidelity and structural separation, and Phase~4 evaluates whether the resulting decomposition satisfies the local contract.

\begin{algorithm}[t]
\caption{SAVED: Realizing Local Decompositionality}
\label{alg:saved}
\begin{algorithmic}[1]
\Require Classifier \(F_\theta : \mathcal{X} \to \mathbb{R}^C\), dataset \(S\), thresholds \((\varepsilon,\kappa,\gamma,\eta)\)
\Ensure Components \(\{m_k\}\), supports \(\{S_k\}\), contract report

\Statex
\Statex \textbf{Stage 1: Structural Decomposition}
\For{$k = 1, \ldots, C$}
    \State $m_k \leftarrow \textsc{DecomposeBinary}(F_\theta, k)$ \Comment{Define class-wise components}
\EndFor

\Statex
\Statex \textbf{Stage 2: Boundary-Aware Input Generation}
\State $S_{\mathrm{bdry}} \leftarrow \textsc{GenerateBoundaryData}(F_\theta, S)$
\Comment{PGD + binary refinement}
\State $S_{\mathrm{cal}} \leftarrow \textsc{Combine}(S, S_{\mathrm{bdry}})$
\Comment{Approximate boundary neighborhood}

\Statex
\Statex \textbf{Stage 3: Semantic-Structural Refinement}
\For{$k = 1, \ldots, C$}
    \State $M_k \leftarrow \textsc{LBMask}(m_k, S_{\mathrm{cal}})$
    \Comment{Learn structured mask}
    \State $m_k \leftarrow \textsc{ApplyMask}(m_k, M_k)$
    \Comment{Refined component}
\EndFor

\Statex
\Statex \textbf{Stage 4: Contract Evaluation}
\State $F_\theta^\downarrow \leftarrow \textsc{Aggregate}(\{m_k\})$
\State Compute $\widehat{\mathrm{Dis}}_\kappa(F_\theta^\downarrow \mid F_\theta; S)$
\State Compute $\mathrm{Overlap}(S_i,S_j)$ and $\mathrm{Size}(S_k)$
\State \Return contract report

\end{algorithmic}
\end{algorithm}

Algorithm~\ref{alg:saved} summarizes the \sysname\ pipeline. 
The procedure follows the abstract definition of decompositionality step by step. 
Stage~1 constructs class-wise components from the original model, defining the structural decomposition. 
Stage~2 generates boundary-aware inputs that approximate the decision boundary on finite data using gradient-based perturbations followed by binary refinement. 
Stage~3 learns structured masks over each component, refining them to preserve boundary-local behavior while inducing structural separation. 
Finally, Stage~4 evaluates the resulting decomposition using empirical measures of semantic disagreement and structural divergence, determining whether the local decompositionality criteria are satisfied. 
This procedure highlights that decompositionality must be realized through boundary-aware learning and subsequently verified through empirical evaluation. Detailed source code is available on our artifact page.

\subsection{Phase 1: Structural Decomposition}
Given \(F_\theta : \mathcal X \to \mathbb R^C\), we construct one binary component per class:
\[
m_k(x)=\bigl[F_{\theta,k}(x), \max_{j\neq k} F_{\theta,j}(x)\bigr].
\]
The aggregated predictor is
\[
\hat y_{F_\theta^\downarrow}(x)=\arg\max_k m_k(x)_0.
\]
This phase defines the \emph{structural} units of decomposition without altering behavior.

\subsection{Phase 2: Boundary-Aware Input Generation}

The abstract definition relies on the boundary neighborhood \(\mathcal N_\tau(F_\theta)\), which is not directly observable. 
\sysname\ approximates this set using data. A natural approach is to use gradient-based adversarial methods (e.g., PGD~\cite{MadryMSTV18}) to generate label-flipping inputs.  However, PGD alone is insufficient. It produces samples that cross the decision boundary but may lie far from it. To obtain tighter boundary approximations, we apply a \emph{binary search refinement} procedure~\cite{KarimiD22}. Given a pair \((x_{\mathrm{lo}}, x_{\mathrm{hi}})\) with different predictions, we iteratively calculate:
\[
x_{\mathrm{mid}} = \tfrac{1}{2}(x_{\mathrm{lo}} + x_{\mathrm{hi}}),
\]
and retain the pair that still straddles the boundary.
This yields samples that lie arbitrarily close to \(\mathcal B(F_\theta)\). The resulting boundary samples are combined with the original dataset to form a calibration set that approximates \(\mathcal N_\tau(F_\theta)\).
This step instantiates the boundary-centric semantic definition discussed in Section~\ref{sec:decomposition}.

\subsection{Phase 3: Semantic-Structural Refinement (LBMask)}

Each component \(m_k\) is pruned via a learned binary mask \(M_k\) applied to the frozen pre-trained weights, i.e., the underlying model parameters are not updated during refinement.

\subsubsection{Objective}
Mask learning is designed to realize local decompositionality by jointly enforcing
(i) semantic fidelity on boundary-relevant inputs and
(ii) structural separation across components.

\subsubsection{Parameterization}
For each layer \(\ell\), we introduce mask logits \(\lambda_\ell\) and obtain mask probabilities via a sigmoid. 
Binary masks are produced during training using a straight-through estimator~\cite{Bengio13}, enabling gradient-based optimization while maintaining discrete structure.

\subsubsection{Structured masking.}
We adopt structured masks~\cite{He24,Li17} (per-neuron or per-channel) to align with the notion of structural support in the decompositionality definition. 
While unstructured masks can preserve predictive behavior, they typically yield entangled sparsity patterns and fail to produce clearly separated components. 
In particular, unstructured pruning often leads to high overlap between component supports, violating the structural disjointness condition (i.e., large \(\mathrm{Overlap}(S_i,S_j)\)), even when predictive performance is maintained. 
As a result, such decompositions fail to satisfy the structural side of the decompositionality contract.

\subsubsection{Architectural realization}
Structured masking removes entire computational units, which can introduce dimensional inconsistencies in architectures with sequential or residual connections~\cite{He16}. 
In particular, pruning output channels in layer~$\ell$ changes the expected input dimension of layer~$\ell{+}1$, and the mismatch propagates through subsequent layers; normalization layers further require consistent slicing of their parameters. 
To address this, we apply a post-hoc \emph{dimension surgery} pass after mask binarization that removes pruned units and propagates the resulting dimensions through the network. 
This includes synchronizing masks across skip connections, slicing adjacent weights and biases, and adjusting normalization parameters. 
After surgery, each component becomes a self-contained sub-network with consistent dimensions. 
This step is essential for CNN-style architectures, whereas in models with self-contained intermediate dimensions (e.g., Transformers), structured masking often yields valid sub-networks without explicit surgery.

\begin{algorithm}[t]
\caption{LBMask (per component \(m_k\))}
\label{alg:lbmask}
\begin{algorithmic}[1]
\Require Component \(m_k\), boundary dataset \(S_{\mathrm{cal}}\)
\Ensure Mask \(M_k\)
\State Initialize mask logits \(\lambda\) (uniform, magnitude-based, etc.)
\For{training steps}
    \State Sample batch \(B \subset S_{\mathrm{cal}}\)
    \State Compute masked forward pass of \(m_k\)
    \State Compute loss: boundary-aware cross-entropy + sparsity regularization
    \State Update \(\lambda\) via SGD
\EndFor
\State Binarize mask to obtain \(M_k\)
\State \Return \(M_k\)
\end{algorithmic}
\end{algorithm}

\subsection{Phase 4: Contract Evaluation}

Given the learned components, we evaluate local decompositionality (Definition~\ref{def:local-decompositionality}).

\paragraph{Semantic fidelity.}
We compute empirical disagreement:
\[
\widehat{\mathrm{Dis}}_\kappa\left(F_\theta^\downarrow \middle| F_\theta; S\right)
=
\frac{1}{|\widehat{\mathcal B}_\kappa|}
\sum_{x\in\widehat{\mathcal B}_\kappa}
\mathbf{1}\left[
\hat y_{F_\theta^\downarrow}(x)\neq \hat y_{F_\theta}(x)
\right].
\]

\paragraph{Structural conditions.}
We measure \(\mathrm{Overlap}(S_i,S_j)\) and \(\mathrm{Size}(S_k)\). A decomposition satisfies local decompositionality if all conditions meet thresholds \((\varepsilon,\kappa,\gamma,\eta)\).


\newcommand{\tbd}[1]{\textcolor{gray}{\texttt{--}}}
\newcommand{\finding}[1]{\medskip\noindent\fbox{\parbox{0.96\linewidth}{\textbf{Finding.} #1}}\medskip}

\section{Evaluation}\label{sec:evaluation}

The goal of our evaluation is not merely to measure pruning quality, but to test a stronger claim: whether a neural network can be decomposed into smaller components that remain \emph{semantically faithful} near decision boundaries while also being \emph{structurally distinct}. These two requirements correspond directly to the two axes of our notion of neural decompositionality. Accordingly, the purpose of the evaluation is not simply to assess whether a method produces smaller subnetworks, but to determine whether it realizes a decomposition that satisfies the semantic-structural contract introduced earlier. To answer this question, we evaluate \sysname\ across three model families, NLP Transformers (BERT), CNNs (ResNet), and ViT (DeiT), and organize the evaluation around three research questions:

\begin{itemize}
    \item \textbf{RQ1:} Can LBMask realize genuine decomposition that satisfies the full local semantic-structural contract?
    \item \textbf{RQ2:} Why are both semantic fidelity and structural separation necessary for certifying decompositionality?
    \item \textbf{RQ3:} How does decompositionality vary across architectures and domains?
\end{itemize}

Across all experiments, the main conclusion is consistent: \emph{decompositionality is neither guaranteed by pruning nor uniformly available across models.} Instead, it is a conditional property that emerges only when boundary-local semantic preservation and structural separation can be achieved simultaneously.

\subsection{Evaluation Setup}\label{sec:exp-setup}

The empirical disagreement $\widehat{\mathrm{Dis}}_\kappa$ (\S\ref{sec:local-semantic-fidelity}) serves as a finite-sample approximation to $\mathrm{Dis}_\tau$; all reported semantic metrics are interpreted as empirical estimates of the abstract boundary-local condition.
We evaluate on representative NLP and vision settings. For NLP, we use BERT-small~\cite{DevlinCLT19} fine-tuned on DBPedia-14 (14-way classification) and AG News~\cite{ZhangZL15} (4-way classification). For vision, we use ResNet-34~\cite{He16} and DeiT-small~\cite{pmlr-v139-touvron21a} on CIFAR-10~\cite{krizhevsky2009learning}. DBPedia-14 is particularly useful because its larger number of classes induces a richer collection of pairwise decision boundaries, making it a stronger stress test for class-specific decomposition.

Unless otherwise noted, we evaluate decompositions under an \((\varepsilon, \kappa, \gamma, \eta)\)-parameterization of the abstract contract, where \(\varepsilon\) controls semantic fidelity, \(\gamma\) bounds structural overlap, \(\eta\) enforces nontrivial reduction, and \(\kappa\) defines the empirical boundary approximation. In our experiments, we set \(\varepsilon = 0.1\), \(\gamma = 0.5\), and \(\eta = 0.3\), and approximate the boundary neighborhood using the \(\kappa\)-quantile with \(q = 0.2\), corresponding to the 20\% lowest-margin samples. This focuses evaluation on boundary-adjacent inputs, where the semantic condition is most critical. We use confidence level \(\delta = 0.05\) for statistical guarantees.

In all LBMask experiments, the original model weights are frozen and only the mask logits are optimized. Our default configuration uses uniform initialization (\(\alpha=0\)), structured masking (per-neuron or per-channel), target sparsity \(s=0.5\), and boundary-aware calibration. Unless explicitly varied, these settings are fixed throughout the evaluation.

We compare against three representative baseline families. First, we use \textbf{Wanda}~\cite{Sun0BK24}, a post-training pruning method based on weight magnitude and activation statistics. Second, we use \textbf{MI pruning}~\cite{Fan21}, which ranks units according to mutual information between activations and labels. Both Wanda and MI are evaluated in structured and/or unstructured variants where applicable, allowing us to separate the effect of pruning granularity from the effect of boundary-awareness. Third, for CNNs we include a prior \textbf{decomposition baseline}~\cite{PanR22} to compare against an existing decomposition-oriented method rather than only against pruning baselines. These baselines play distinct roles in the evaluation: some preserve predictions while collapsing structure, whereas others achieve structural reduction while destroying semantics.

For each target class \(k\), we construct a binary component \(m_k\), apply LBMask or a baseline pruning or decomposition method, aggregate the resulting components into \(F_\theta^\downarrow\), and evaluate the final decomposition against the local contract. Concretely, the semantic condition is evaluated through \(\widehat{\mathrm{Dis}}_\kappa\), which serves as the empirical estimator of \(\mathrm{Dis}_\tau\); the structural condition is evaluated through the maximum pairwise overlap \(\max_{i\neq j}\mathrm{Overlap}(S_i,S_j)\); and the nontriviality of the decomposition is evaluated through the minimum component reduction, reported by Prune or an equivalent retained-size ratio. We additionally report Hoeffding correction terms and boundary-restricted confusion-matrix deviation when relevant. A decomposition is counted as satisfying neural decompositionality only if \emph{all} of these conditions hold simultaneously.

To provide a high-level overview before examining each research question in detail, Table~\ref{tab:eval-overview} summarizes which representative method--architecture settings satisfy each axis of the proposed decompositionality contract. The table should be read as a map of failure modes: some settings satisfy the semantic condition but fail structurally, while others achieve structural sparsification but violate the boundary-local semantic condition. A setting constitutes genuine decompositionality only when it satisfies semantic fidelity, structural separation, and nontrivial reduction simultaneously.

\subsection{Overview of Evaluation Results}\label{sec:overview-experimental results}

\begin{table}[t]
    \centering
    \small
    \begin{tabular}{lccc}
        \toprule
        \textbf{Setting} & \textbf{Semantic} & \textbf{Structural} & \textbf{Contract} \\
        \midrule
        BERT / DBPedia + LBMask (\(\alpha=0\)) & \checkmark & \checkmark & \checkmark \\
        BERT / AG News + LBMask & $\times$ & $\checkmark$ & $\times$ \\
        Wanda (unstructured) & \checkmark & $\times$ & $\times$ \\
        MI pruning (structured) & $\times$ & $\times$ & $\times$ \\
        Vision (ResNet/DeiT) + LBMask & $\times$ & \checkmark & $\times$ \\
        Vision + unstructured masking & partial & $\times$ & $\times$ \\
        \bottomrule
    \end{tabular}
    \caption{
    \textbf{Decompositionality outcomes.}
    Only LBMask on BERT/DBPedia satisfies the full contract.
    Other methods fail either semantically or structurally.
    }
    \label{tab:eval-overview}    
\end{table}

Table~\ref{tab:eval-overview} summarizes the key empirical pattern of our evaluation. Only one configuration satisfies the full contract, while all others fail in systematically different ways. These failures are not uniform: some methods preserve boundary-level semantics but collapse structurally, whereas others achieve structural reduction while violating the semantic condition. The remainder of the evaluation analyzes this pattern in detail. RQ1 establishes the realizability of the contract, RQ2 demonstrates the necessity of both axes, and RQ3 characterizes its dependence on architecture.

\subsection{Answer to RQ1: Realizability of Neural Decompositionality}

\begin{table}[t]
    \centering
    \small
    \begin{tabular}{lcccccc}
        \toprule
        \textbf{Dataset}
        & $\widehat{Dis}_\kappa$
        & $\widehat{Dis}_\kappa + h \le \varepsilon$
        & $\max \mathrm{Overlap} \le \gamma$
        & $\min \mathrm{Prune} \ge \eta$
        & $\|C^\downarrow_\kappa - C_\kappa\|_1$
        & \textbf{Contract} \\
        \midrule
        DBPedia-14
        & 0.0521
        & $0.0636 \le 0.1$ \checkmark
        & $0.3629 \le 0.5$ \checkmark
        & $0.5000 \ge 0.3$ \checkmark
        & 0.0700
        & \checkmark \\
        AG News
        & 0.2125
        & $0.2473 > 0.1$ $\times$
        & $0.3950 \le 0.5$ \checkmark
        & $0.5000 \ge 0.3$ \checkmark
        & 0.1697
        & $\times$ \\
        \bottomrule
    \end{tabular}
    \caption{
    \textbf{RQ1: Contract-level evaluation of LBMask under $(\varepsilon,\kappa,\gamma,\eta)$.}
    Each column corresponds to a component of the abstract decompositionality contract.
    A decomposition satisfies the contract only if all conditions hold simultaneously.
    }
    \label{tab:rq1-main}    
\end{table}

Table~\ref{tab:rq1-main} presents the contract-level evaluation of LBMask on two NLP tasks under the $(\varepsilon, \kappa, \gamma, \eta)$-parameterization.
Recall that a decomposition satisfies neural decompositionality only if it simultaneously meets all three conditions:
(i) semantic fidelity, $\mathrm{Dis}_\kappa(F_\theta^\downarrow \mid F_\theta) \le \varepsilon$;
(ii) structural separation, $\max_{i \neq j}\mathrm{Overlap}(S_i,S_j) \le \gamma$; and
(iii) nontrivial reduction, $\min_i \mathrm{Prune}(S_i) \ge \eta$.
In practice, the semantic condition is evaluated via the empirical estimator $\widehat{Dis}_\kappa$, which approximates boundary-local disagreement over the $\kappa$-restricted subset.

\subsubsection{DBPedia-14: realization of the full contract.}
On DBPedia-14 ($K{=}14$), LBMask satisfies all components of the contract.
The aggregated boundary disagreement is $\widehat{Dis}_\kappa = 0.0521$, which lies well below the semantic threshold $\varepsilon = 0.1$.
With Hoeffding correction $h = 0.0115$, we obtain
\[
\widehat{Dis}_\kappa + h = 0.0636 \le \varepsilon,
\]
establishing semantic fidelity with $95\%$ confidence.
At the same time, the structural conditions are satisfied.
The maximum pairwise overlap is $0.3629 < \gamma = 0.5$, indicating that the learned components are not redundant,
and every component achieves the target pruning ratio $0.5000 \ge \eta = 0.3$, confirming that the decomposition is nontrivial.

Beyond these primary metrics, the confusion-matrix deviation
$\|C^\downarrow_\kappa - C_\kappa\|_1 = 0.0700$ is well within the theoretical bound $2\varepsilon = 0.2$ predicted by the Main Theorem.
This consistency provides additional evidence that boundary-local semantic guarantees translate into stable multiclass behavior.
Taken together, these results demonstrate that LBMask produces a decomposition that is simultaneously semantically faithful, structurally distinct, and nontrivially reduced.

\subsubsection{Per-class heterogeneity and boundary complexity.}
Although the aggregated disagreement is low, the per-class results reveal substantial heterogeneity.
The best class (class~10) achieves $\widehat{Dis}_\kappa^{(k)} = 0.0165$, the median class (class~4) achieves $0.0675$, while the worst class (class~8) reaches $0.5072$.
This variation reflects the non-uniform geometry of decision boundaries.
Certain classes are surrounded by semantically similar neighbors, resulting in dense clusters of low-margin inputs where even boundary-aware pruning struggles to preserve predictions.
In contrast, more isolated classes admit simpler local decision regions that are easier to approximate under sparsification.

Importantly, however, neural decompositionality is defined at the level of the aggregated classifier $F_\theta^\downarrow$, not at the level of individual binary components.
The final prediction is obtained via $\arg\max$ aggregation across all $K$ components, which allows errors in individual components to be compensated by others.
As a result, even though some classes exhibit high local disagreement, the overall boundary-level behavior remains stable, and the global semantic condition is satisfied.

\subsubsection{AG News: structural success but semantic failure.}
In contrast, the results on AG News ($K{=}4$) show that decompositionality is not guaranteed.
Here, LBMask achieves comparable structural properties:
the maximum overlap is $0.3950 < \gamma$, and all components satisfy the pruning requirement ($0.5000 \ge \eta$).
However, the semantic condition is violated:
the aggregated disagreement $\widehat{Dis}_\kappa = 0.2125$ exceeds $\varepsilon = 0.1$, and all per-class disagreements fall within $[0.18, 0.31]$.

Unlike DBPedia-14, this degradation is uniform across classes, suggesting a task-level limitation rather than a single difficult boundary.
One possible explanation is that with fewer classes, each decision boundary must encode a broader portion of the semantic space.
At fixed sparsity, this increases the representational burden per component, making it harder to preserve all boundary regions simultaneously.
However, we emphasize that AG News and DBPedia-14 differ in several confounding factors beyond class count, including domain, vocabulary, and semantic granularity.
A controlled study would be required to isolate the precise cause of this behavior.

It is also worth noting that the confusion-matrix deviation $\|C^\downarrow_\kappa - C_\kappa\|_1 = 0.1697$ is reported for completeness, but the theoretical bound $\le 2\varepsilon$ does not apply in this case because the semantic condition is violated.
This further highlights the central role of boundary-local semantic fidelity in ensuring global stability.

\subsubsection{Conclusion.}
Taken together, these results provide a clear answer to RQ1.
\emph{Neural decompositionality is a realizable but conditional property.}
The success on DBPedia-14 demonstrates that a neural network can be decomposed into multiple structurally distinct components that collectively preserve boundary-local semantics.
At the same time, the failure on AG News shows that this property is not inherent to the decomposition algorithm alone, but depends on whether the underlying model admits a boundary-preserving factorization at the given sparsity level.

\subsection{Answer to RQ2: Why Both Metrics Are Necessary}\label{sec:rq2}

To answer RQ2, we examine whether semantic fidelity ($\widehat{Dis}_\kappa$) and structural separation (Overlap) can be satisfied independently, and whether either metric alone is sufficient to certify a valid decomposition. Table~\ref{tab:rq2-main} reveals two distinct and complementary failure modes, showing that the two metrics capture fundamentally different properties.

\begin{table}[t]
    \centering
    \small
    \begin{tabular}{lccc}
        \toprule
        \textbf{Method} & $\widehat{Dis}_\kappa$ & \textbf{Overlap} & \textbf{Failure mode} \\
        \midrule
        LBM-U ($\alpha{=}0$) & 0.0521 & 0.3629 & -- (satisfies both) \\
        Wanda (unstr.)       & 0.0097 & 0.9940 & structural collapse \\
        Wanda (str.)         & 0.3795 & 0.9885 & semantic + structural failure \\
        MI (str.)            & 0.3491 & 0.8668 & semantic failure \\
        LBM-M ($\alpha{=}1$) & 0.1272 & 0.9223 & high overlap \\
        \bottomrule
    \end{tabular}
    \caption{
    \textbf{RQ2: Metric necessity via failure modes (DBPedia-14, $s{=}0.5$, $\kappa$-quantile $q{=}0.2$).}
    Each method fails either semantic fidelity or structural separation.
    }
    \label{tab:rq2-main}    
\end{table}

\subsubsection{Semantic preservation does not imply decomposition.}
Wanda-unstructured achieves the lowest disagreement among all methods ($\widehat{Dis}_\kappa = 0.0097$), indicating near-perfect boundary-level semantic fidelity.
However, its overlap is $0.9940$, meaning that all class-wise components share nearly identical supports.
This corresponds to a degenerate solution in which the model preserves predictions by reusing the same subnetwork for every class, producing no meaningful structural separation.
Thus, semantic fidelity alone is insufficient. It admits trivial solutions that collapse the decomposition.

\subsubsection{Structural sparsification does not preserve semantics.}
Structured baselines exhibit the opposite failure.
Wanda-structured and MI-structured both achieve nontrivial pruning but suffer from large disagreement ($0.3795$ and $0.3491$, respectively), indicating severe semantic distortion near decision boundaries.
These methods select globally important neurons shared across classes, rather than features that distinguish classes locally.
As a result, structural reduction alone does not preserve the classifier's decision-boundary behavior.

\subsubsection{Boundary-aware optimization requires structural flexibility.}
LBMask with magnitude initialization ($\alpha{=}1$) partially reduces disagreement ($0.1272$), but still exhibits high overlap ($0.9223$).
The magnitude prior anchors all masks to shared high-norm neurons, preventing structural separation.
Only LBMask with uniform initialization ($\alpha{=}0$) simultaneously achieves low disagreement ($0.0521$) and low overlap ($0.3629$), satisfying both conditions.

\subsubsection{Conclusion.}
These results show that semantic fidelity and structural separation are both necessary to characterize neural decompositionality.
Semantic fidelity eliminates structurally collapsed solutions, while structural separation eliminates semantically invalid decompositions.
\emph{Neither condition alone is sufficient: only their conjunction yields a meaningful decomposition that preserves boundary-level behavior while producing distinct components.}

\subsection{RQ3: Architecture-Level Decompositionality}\label{sec:rq3}

RQ1 established that neural decompositionality can be realized in practice, and RQ2 showed that both semantic fidelity and structural separation are necessary.
RQ3 asks a complementary question: \emph{where} does the $(\varepsilon,\kappa,\gamma,\eta)$-contract hold?
In particular, does decompositionality emerge uniformly across architectures, or do certain model families exhibit intrinsic barriers?
Table~\ref{tab:rq3-summary} summarizes the central architectural pattern.
Only BERT satisfies both semantic and structural conditions simultaneously, while both CNN and ViT fail the semantic condition despite achieving structural sparsification.

\begin{table}[t]
    \centering
    \small
    \begin{tabular}{lccc}
        \toprule
        \textbf{Architecture} & \textbf{Semantic} & \textbf{Structural} & \textbf{Contract} \\
        \midrule
        BERT (NLP Transformer) & \checkmark & \checkmark & \checkmark \\
        CNN (ResNet)           & $\times$      & \checkmark & $\times$ \\
        ViT (DeiT)             & $\times$      & \checkmark & $\times$ \\
        \bottomrule
    \end{tabular}
    \caption{
    \textbf{RQ3: Decompositionality across architectures.}
    Only BERT satisfies the full contract.
    Vision models exhibit a systematic trade-off between semantic fidelity and structural separation.
    }
    \label{tab:rq3-summary}
\end{table}

\subsubsection{BERT: Contract satisfaction under structured sparsification.}
For BERT, the $(\varepsilon,\kappa,\gamma,\eta)$-contract is satisfied at moderate sparsity.
Boundary disagreement remains below the semantic threshold while overlap decreases as sparsity increases, yielding structurally distinct components without sacrificing semantic fidelity.
Uniform initialization ($\alpha{=}0$) is critical in enabling class-specific mask discovery.
Overall, BERT admits a boundary-preserving and structurally separated decomposition.

\subsubsection{CNN: Persistent semantic failure.}
ResNet fails to satisfy the semantic condition across all configurations.
While increasing sparsity improves structural separation, boundary disagreement remains consistently high.
Structured masking preserves structure but destroys semantics, whereas unstructured masking can preserve semantics only at the cost of structural collapse.
No configuration satisfies both contract axes simultaneously.

\subsubsection{ViT: Partial modularity without contract satisfaction.}
DeiT exhibits intermediate behavior between BERT and CNN.
Structured masking yields lower disagreement than CNN but still fails the semantic threshold.
Unstructured masking can recover semantic fidelity, but again only with near-total overlap.
Thus, ViT shows limited modularity but still fails to satisfy the full contract.

\subsubsection{A consistent trade-off in vision models.}
Across both CNN and ViT, all configurations follow the same pattern:
improving semantic fidelity increases overlap, while improving structural separation increases semantic error.
This trade-off persists across sparsity levels, masking granularities, and initialization schemes.

\subsubsection{Interpretation.}
These results point to a fundamental architectural distinction.
BERT encodes class-discriminative information in sparse, class-specific components, enabling decomposition.
In contrast, vision models rely on distributed representations, where class information is shared across many features.
This makes it difficult to isolate class-specific components without disrupting boundary semantics.

\subsubsection{Conclusion.}
Neural decompositionality is therefore architecture-dependent.
It emerges in NLP Transformers but fails in CNNs and ViTs under all tested configurations.
This indicates that decompositionality is governed by representation structure rather than the decomposition method itself.

\subsection{Discussion}\label{sec:discussion}

\subsubsection{Decompositionality is realizable but conditional.}
Our evaluation shows that neural decompositionality is not a byproduct of pruning or compression, but a property that must be explicitly realized and verified.
RQ1 demonstrates that the $(\varepsilon,\kappa,\gamma,\eta)$-contract can be satisfied in practice: LBMask produces decompositions that are both semantically faithful and structurally distinct on BERT.
However, the same method fails on AG News, indicating that decompositionality depends on the interaction between model, task, and representational capacity rather than on the decomposition algorithm alone.

\subsubsection{Both contract axes are necessary.}
RQ2 establishes that semantic fidelity and structural separation capture orthogonal failure modes.
Methods that optimize only semantic fidelity collapse structurally, producing degenerate decompositions with nearly identical components.
Conversely, methods that optimize only structural sparsity destroy boundary-level semantics.
Thus, neither axis alone is sufficient: meaningful decomposition requires their conjunction.
This validates the need for a two-dimensional contract rather than a single aggregate metric.

\subsubsection{Architecture determines decompositionality.}
RQ3 reveals a sharp architectural divide.
BERT satisfies the full contract at moderate sparsity, indicating that its representations admit a boundary-preserving factorization.
In contrast, both CNNs and ViT fail under all tested configurations.
Across these models, a consistent trade-off emerges: preserving boundary semantics requires retaining shared features across classes, while enforcing structural separation disrupts those same features.
This suggests that decompositionality is governed by how class-discriminative information is distributed---sparse and separable in NLP models, but dense and entangled in vision models.

\subsubsection{Implications for neural modularity.}
These findings suggest that neural decompositionality should be understood as a property of representation rather than of decomposition method.
Architectures that encode class-specific information in localized substructures naturally support modular decomposition, while architectures that rely on distributed representations do not.
This perspective aligns neural decomposition with classical notions of software modularity, where meaningful modules correspond to separable units of functionality.

\subsubsection{Limitations and future directions.}
Our study focuses on a fixed contract threshold and a limited set of models and tasks.
While the qualitative patterns are robust, further work is needed to explore larger-scale models, alternative architectures, and different operating regimes.
In particular, relaxing the contract or adopting multi-objective formulations may reveal intermediate regimes of partial decompositionality.
Additionally, evaluating the usefulness of decomposed components in downstream tasks such as verification or modular reuse remains an important direction for future work.

\section{Related Works}\label{sec:related-works}

A large body of works has investigated on relevant topics, including decision boundary anlaysis, testing approaches, network decomposition via pruning, and formal verifier of decision boundaries, etc.

\paragraph{\textbf{Decision-boundary analysis and boundary-oriented input generation.}}

A large body of work studies input generation techniques that probe or manipulate the decision boundary of neural networks.
Adversarial example generation methods such as FGSM~\cite{GoodfellowSS15}, PGD~\cite{MadryMSTV18}, DeepFool~\cite{Moosavi-Dezfooli16}, and C\&W~\cite{Carlini017} aim to find perturbations that cross the decision boundary while remaining close to the original input.
These methods demonstrate that decision-boundary structure is central to the behavior of neural classifiers, but their primary goal is to expose vulnerability or generate adversarial counterexamples rather than to characterize modularity or decomposition.

Related efforts have explored boundary-adjacent input generation more directly.
For example, DeepDIG~\cite{KarimiD22} generates inputs close to decision boundaries by constructing adversarial examples and refining them via binary search.
Other works study decision-boundary geometry, perturbation sensitivity, or visualization of class transitions~\cite{Oyallon17,SomepalliFBCDBG22,KarimiTD19,FawziMFS17,ZhaoNG24}.
These methods are highly relevant to our work because they reinforce the importance of boundary regions.
However, they do not define decompositionality as a semantic property, nor do they use boundary behavior as a formal contract for validating network decomposition.

Our work differs in two ways.
First, we elevate the decision boundary from an analysis target to a \emph{semantic reference structure}: decomposition is considered valid only when it preserves classifier behavior near this boundary.
Second, we use boundary-oriented inputs not to construct attacks, but to operationalize \emph{boundary-aware semantic fidelity} and its local approximation.
Thus, unlike prior boundary-probing methods, our framework turns boundary sensitivity into a formal criterion for decomposition.

\paragraph{\textbf{Coverage-guided testing and input generation.}}

Another line of research focuses on input generation for testing and coverage maximization~\cite{OdenaOAG19,8835342}.
These approaches attempt to activate neurons, layers, or internal behaviors that would otherwise remain untested.
Although such methods are useful for systematic exploration of neural-network behavior, they largely define adequacy in syntactic terms, such as neuron coverage or activation diversity.
Our work is complementary but distinct.
Rather than maximizing structural coverage, we focus on semantically critical regions near the decision boundary.
In this sense, our notion of local semantic fidelity is closer to boundary-sensitive behavioral preservation than to conventional testing adequacy criteria.

\paragraph{\textbf{Modular architectures and mixture-of-experts.}}
Neural modularity has also been explored through architectural design,
including mixture-of-experts models~\cite{JacobsJNH91} and routing-based networks~\cite{10.5555/3294996.3295142}, where
different components are trained to specialize on subsets of the input space.
These approaches introduce modularity during training via explicit architectural
constraints.

In contrast, our work studies \emph{post hoc} decomposition of a trained model
without modifying its architecture or retraining.
Moreover, while modular architectures encourage specialization,
they do not provide a formal criterion for when the resulting components
constitute a semantically valid decomposition of the original model.
Our formulation addresses this gap by defining decompositionality as a
boundary-aware semantic--structural contract.



    
\paragraph{\textbf{Network pruning and decomposition.}}

Our work is closely related to pruning, model compression, and modularization, but differs fundamentally in objective and formalization.
Most pruning methods, including both unstructured and structured pruning techniques~\cite{Cheng24}, aim to reduce model size or computational cost while preserving aggregate predictive accuracy.
Methods such as Wanda~\cite{Sun0BK24} compute importance scores based on weights and activations, while in-training pruning approaches such as HYDRA~\cite{SehwagWMJ20} and related robustness-aware pruning frameworks optimize sparse masks jointly with model parameters.
These methods are highly effective for compression, but they do not define when a pruned or partitioned model constitutes a valid \emph{decomposition} in a semantic sense.

The distinction is central to our formulation.
In our framework, a decomposition must satisfy two conditions: (1) \emph{semantic fidelity}, expressed as low disagreement near the decision boundary, and (2) \emph{structural divergence}, expressed through mask disjointness and non-trivial pruning.
This differs from standard pruning, where two subnetworks may retain high accuracy while still sharing nearly identical internal structure, and thus fail to provide meaningful modular separation.

Prior work has also explicitly studied neural network decomposition.
Pan and Rajan~\cite{PanR20} decompose monolithic DNNs into smaller concern-based modules, and later extend this line of work to CNNs~\cite{PanR22}; Imtiaz et al.~\cite{ImtiazBSPCR23} further adapt similar ideas to recurrent models.
These works are among the closest to ours in spirit, as they seek to recover modularity in neural networks.
However, their decompositions are primarily guided by concern identification and structural factorization, without a formal semantic contract based on classifier behavior near class-transition regions.
Our work differs by defining decompositionality itself as a joint semantic--structural property: preserving boundary-level behavior while ensuring that the resulting components are non-trivially distinct.

\paragraph{\textbf{Neural network verification and robustness analysis.}}

A substantial literature studies formal verification of neural networks, especially robustness verification.
Existing approaches are commonly divided into constraint-based methods and abstraction-based methods~\cite{HuangKRSSTWY20,albarghouthi2021introduction}, with tools such as Marabou~\cite{KatzHIJLLSTWZDK19,WuIZTDKRAJBHLWZKKB24}, $\alpha,\beta$-CROWN~\cite{zhang2018efficient,xu2020automatic,salman2019convex,xu2021fast,wang2021beta}, and ERAN~\cite{SinghGMPV18,SinghGPV19}.
These systems typically reason about local robustness or safety properties of a fixed model.
Their focus is not on whether a model admits a modular decomposition that preserves semantics.

Nevertheless, our work is closely connected to this line of research in two important ways.
First, our emphasis on the decision boundary is directly inspired by robustness verification, where small-margin or adversarially vulnerable inputs mark semantically unstable regions.
Second, our transition from global decompositionality to local decompositionality parallels the distinction between global and local robustness, and our dataset-based abstraction is conceptually informed by abstract interpretation.
Unlike standard verification frameworks, however, we do not aim to certify all perturbations in an input region; instead, we use local robustness reasoning as a \emph{semantic probe} for whether decomposition preserves class-transition behavior.

Several recent works aim to improve the scalability of verification by reusing proofs or latent abstractions across related models~\cite{Marc22,Shubham22}, or by extracting subnetworks connected to latent representations for separate verification~\cite{HanspalL23}.
These efforts are highly relevant because they show that structural reduction can improve verification scalability.
However, they do not ask when such reductions are semantically justified as decompositions of the original model.
Our work addresses precisely this gap by providing a formal definition of decompositionality and an empirical framework for evaluating it without retraining the original model.

\paragraph{\textbf{Abstraction and local reasoning.}}

The abstraction used in our local formulation is inspired by two neighboring traditions.
From neural network verification, we inherit the global-versus-local distinction: a semantic property stated over the entire input space is often intractable, while a local approximation around critical regions is analyzable in practice.
From abstract interpretation, we inherit the idea that an abstract domain can provide a tractable surrogate for an intractable concrete semantics.
Our margin-based boundary abstraction follows this spirit by replacing the true boundary neighborhood with a finite subset of low-margin inputs.

At the same time, our abstraction differs from standard sound abstractions in program analysis.
We do not claim that the margin-induced boundary subset is a sound over-approximation of the full decision boundary.
Instead, it is a practical and semantically motivated approximation designed to preserve the intent of decompositionality while enabling empirical realization.
To our knowledge, prior work has not formulated neural network decomposition through such a boundary-centric semantic abstraction.

\section{Conclusion and Future Work} \label{sec:conclusion}

\paragraph{\textbf{Conclusion.}}
We introduced \emph{neural decompositionality}, a formally defined
property that characterizes when a trained neural network admits a
semantically meaningful decomposition. Unlike prior work that relies on
aggregate accuracy, our formulation grounds decompositionality in the
classifier's decision boundary, where semantic transitions occur. We
formalize decompositionality as a joint semantic-structural contract
consisting of (i)~boundary-aware semantic fidelity and (ii)~structural
divergence, ensuring both behavioral preservation and non-trivial
modular separation.

To operationalize this definition, we developed \sysname, a
boundary-aware decomposition framework that combines boundary-focused
input generation with LBMask. The framework concentrates on optimization
in decision-boundary regions while preserving the original model
parameters, enabling the discovery of class-specific structural
components driven by boundary-relevant gradients.

Our empirical study shows that decompositionality is achievable, but
highly non-trivial. The proposed contract exposes distinct failure modes
that are not captured by standard evaluation. For example, unstructured methods
preserve boundary-level semantics but fail to produce meaningful
structure, whereas class-agnostic structured methods collapse both
semantic fidelity and structural divergence. Furthermore,
decompositionality exhibits a clear architectural and task bias. NLP Transformers
are the most amenable. However, Vision Transformers partially satisfy semantic
constraints but fail to meet the full contract, and CNNs largely resist
decomposition due to the diffuse nature of class-discriminative
computation. These results establish decompositionality as a principled
and empirically testable notion, providing a foundation for modular
reasoning over neural networks.

\paragraph{\textbf{Future work.}}
Our concept is new, and several directions remain open.
First, scaling to larger Transformer architectures and more diverse
tasks is necessary to assess the generality of the observed NLP--vision
divide. Second, adaptive per-class sparsity allocation may improve
contract satisfaction in settings where uniform sparsity is insufficient.
Third, establishing a connection between decompositionality and
verification is a key next step: if components can be verified
independently and the contract guarantees boundary-level semantic
preservation under composition, then verification cost should scale with
component size rather than that of the full model.

More fundamentally, our results suggest that decompositionality is not
only a property of the decomposition method but also of the model,
namely how class-discriminative computation is organized during training.
This motivates \emph{decomposition-aware training}, where objectives or
architectural constraints encourage modular representations amenable to
post-hoc decomposition. Possible directions include penalizing
cross-class neuron co-activation, enforcing modular bottlenecks, or using
multi-task pretraining to induce role separation. Such approaches would
elevate decompositionality from a post-hoc diagnostic to a design
principle, potentially extending it to vision models where current
methods fail.

Finally, inference-time architectural interventions may complement
training-time strategies by introducing structural separability into
otherwise entangled architectures. Such interventions can be viewed as
instantiations of the same principle, including mechanisms such as
class-conditional attention routing and sparse mixture-of-experts.

\bibliographystyle{ACM-Reference-Format}
\bibliography{refs}


\begin{thebibliography}{54}


\ifx \showCODEN    \undefined \def \showCODEN     #1{\unskip}     \fi
\ifx \showDOI      \undefined \def \showDOI       #1{#1}\fi
\ifx \showISBNx    \undefined \def \showISBNx     #1{\unskip}     \fi
\ifx \showISBNxiii \undefined \def \showISBNxiii  #1{\unskip}     \fi
\ifx \showISSN     \undefined \def \showISSN      #1{\unskip}     \fi
\ifx \showLCCN     \undefined \def \showLCCN      #1{\unskip}     \fi
\ifx \shownote     \undefined \def \shownote      #1{#1}          \fi
\ifx \showarticletitle \undefined \def \showarticletitle #1{#1}   \fi
\ifx \showURL      \undefined \def \showURL       {\relax}        \fi
\providecommand\bibfield[2]{#2}
\providecommand\bibinfo[2]{#2}
\providecommand\natexlab[1]{#1}
\providecommand\showeprint[2][]{arXiv:#2}

\bibitem[Albarghouthi(2021)]%
        {albarghouthi2021introduction}
\bibfield{author}{\bibinfo{person}{Aws Albarghouthi}.}
  \bibinfo{year}{2021}\natexlab{}.
\newblock \showarticletitle{Introduction to Neural Network Verification}.
\newblock \bibinfo{journal}{\emph{Foundations and Trends in Programming
  Languages}} \bibinfo{volume}{7}, \bibinfo{number}{1-2} (\bibinfo{date}{12}
  \bibinfo{year}{2021}), \bibinfo{pages}{1--157}.
\newblock
\showISSN{2325-1107}
\urldef\tempurl%
\url{https://doi.org/10.1561/2500000051}
\showDOI{\tempurl}
\showeprint{https://www.emerald.com/ftpgl/article-pdf/7/1-2/1/11044581/2500000051en.pdf}


\bibitem[Bengio et~al\mbox{.}(2013)]%
        {Bengio13}
\bibfield{author}{\bibinfo{person}{Yoshua Bengio}, \bibinfo{person}{Nicholas
  L{\'{e}}onard}, {and} \bibinfo{person}{Aaron~C. Courville}.}
  \bibinfo{year}{2013}\natexlab{}.
\newblock \showarticletitle{Estimating or Propagating Gradients Through
  Stochastic Neurons for Conditional Computation}.
\newblock \bibinfo{journal}{\emph{CoRR}}  \bibinfo{volume}{abs/1308.3432}
  (\bibinfo{year}{2013}).
\newblock
\showeprint[arXiv]{1308.3432}
\urldef\tempurl%
\url{http://arxiv.org/abs/1308.3432}
\showURL{%
\tempurl}


\bibitem[Carlini and Wagner(2017)]%
        {Carlini017}
\bibfield{author}{\bibinfo{person}{Nicholas Carlini} {and}
  \bibinfo{person}{David~A. Wagner}.} \bibinfo{year}{2017}\natexlab{}.
\newblock \showarticletitle{Towards Evaluating the Robustness of Neural
  Networks}. In \bibinfo{booktitle}{\emph{2017 {IEEE} Symposium on Security and
  Privacy, {SP} 2017, San Jose, CA, USA, May 22-26, 2017}}.
  \bibinfo{publisher}{{IEEE} Computer Society}, \bibinfo{pages}{39--57}.
\newblock
\urldef\tempurl%
\url{https://doi.org/10.1109/SP.2017.49}
\showDOI{\tempurl}


\bibitem[Cheng et~al\mbox{.}(2024)]%
        {Cheng24}
\bibfield{author}{\bibinfo{person}{Hongrong Cheng}, \bibinfo{person}{Miao
  Zhang}, {and} \bibinfo{person}{Javen~Qinfeng Shi}.}
  \bibinfo{year}{2024}\natexlab{}.
\newblock \showarticletitle{A Survey on Deep Neural Network Pruning: Taxonomy,
  Comparison, Analysis, and Recommendations}.
\newblock \bibinfo{journal}{\emph{IEEE Transactions on Pattern Analysis and
  Machine Intelligence}} \bibinfo{volume}{46}, \bibinfo{number}{12}
  (\bibinfo{year}{2024}), \bibinfo{pages}{10558--10578}.
\newblock
\urldef\tempurl%
\url{https://doi.org/10.1109/TPAMI.2024.3447085}
\showDOI{\tempurl}


\bibitem[Devlin et~al\mbox{.}(2019)]%
        {DevlinCLT19}
\bibfield{author}{\bibinfo{person}{Jacob Devlin}, \bibinfo{person}{Ming{-}Wei
  Chang}, \bibinfo{person}{Kenton Lee}, {and} \bibinfo{person}{Kristina
  Toutanova}.} \bibinfo{year}{2019}\natexlab{}.
\newblock \showarticletitle{{BERT:} Pre-training of Deep Bidirectional
  Transformers for Language Understanding}. In
  \bibinfo{booktitle}{\emph{Proceedings of the 2019 Conference of the North
  American Chapter of the Association for Computational Linguistics: Human
  Language Technologies, {NAACL-HLT} 2019, Minneapolis, MN, USA, June 2-7,
  2019, Volume 1 (Long and Short Papers)}},
  \bibfield{editor}{\bibinfo{person}{Jill Burstein}, \bibinfo{person}{Christy
  Doran}, {and} \bibinfo{person}{Thamar Solorio}} (Eds.).
  \bibinfo{publisher}{Association for Computational Linguistics},
  \bibinfo{pages}{4171--4186}.
\newblock


\bibitem[Fan et~al\mbox{.}(2021)]%
        {Fan21}
\bibfield{author}{\bibinfo{person}{Chun Fan}, \bibinfo{person}{Jiwei Li},
  \bibinfo{person}{Tianwei Zhang}, \bibinfo{person}{Xiang Ao},
  \bibinfo{person}{Fei Wu}, \bibinfo{person}{Yuxian Meng}, {and}
  \bibinfo{person}{Xiaofei Sun}.} \bibinfo{year}{2021}\natexlab{}.
\newblock \showarticletitle{Layer-wise Model Pruning based on Mutual
  Information}. In \bibinfo{booktitle}{\emph{Proceedings of the 2021 Conference
  on Empirical Methods in Natural Language Processing, {EMNLP} 2021}}.
  \bibinfo{publisher}{Association for Computational Linguistics},
  \bibinfo{pages}{3079--3090}.
\newblock


\bibitem[Fawzi et~al\mbox{.}(2017)]%
        {FawziMFS17}
\bibfield{author}{\bibinfo{person}{Alhussein Fawzi},
  \bibinfo{person}{Seyed{-}Mohsen Moosavi{-}Dezfooli}, \bibinfo{person}{Pascal
  Frossard}, {and} \bibinfo{person}{Stefano Soatto}.}
  \bibinfo{year}{2017}\natexlab{}.
\newblock \showarticletitle{Classification regions of deep neural networks}.
\newblock \bibinfo{journal}{\emph{CoRR}}  \bibinfo{volume}{abs/1705.09552}
  (\bibinfo{year}{2017}).
\newblock
\showeprint[arXiv]{1705.09552}
\urldef\tempurl%
\url{http://arxiv.org/abs/1705.09552}
\showURL{%
\tempurl}


\bibitem[Fischer et~al\mbox{.}(2022)]%
        {Marc22}
\bibfield{author}{\bibinfo{person}{Marc Fischer}, \bibinfo{person}{Christian
  Sprecher}, \bibinfo{person}{Dimitar~Iliev Dimitrov},
  \bibinfo{person}{Gagandeep Singh}, {and} \bibinfo{person}{Martin Vechev}.}
  \bibinfo{year}{2022}\natexlab{}.
\newblock \showarticletitle{Shared Certificates for Neural Network
  Verification}. In \bibinfo{booktitle}{\emph{Computer Aided Verification}},
  \bibfield{editor}{\bibinfo{person}{Sharon Shoham} {and}
  \bibinfo{person}{Yakir Vizel}} (Eds.). \bibinfo{publisher}{Springer
  International Publishing}, \bibinfo{address}{Cham},
  \bibinfo{pages}{127--148}.
\newblock
\showISBNx{978-3-031-13185-1}


\bibitem[Frankle and Carbin(2019)]%
        {frankle2019}
\bibfield{author}{\bibinfo{person}{Jonathan Frankle} {and}
  \bibinfo{person}{Michael Carbin}.} \bibinfo{year}{2019}\natexlab{}.
\newblock \showarticletitle{The Lottery Ticket Hypothesis: Finding Sparse,
  Trainable Neural Networks}. In \bibinfo{booktitle}{\emph{7th International
  Conference on Learning Representations, {ICLR} 2019, New Orleans, LA, USA,
  May 6-9, 2019}}. \bibinfo{publisher}{OpenReview.net}.
\newblock
\urldef\tempurl%
\url{https://openreview.net/forum?id=rJl-b3RcF7}
\showURL{%
\tempurl}


\bibitem[Goodfellow et~al\mbox{.}(2016)]%
        {Goodfellow-et-al-2016}
\bibfield{author}{\bibinfo{person}{Ian Goodfellow}, \bibinfo{person}{Yoshua
  Bengio}, {and} \bibinfo{person}{Aaron Courville}.}
  \bibinfo{year}{2016}\natexlab{}.
\newblock \bibinfo{booktitle}{\emph{Deep Learning}}.
\newblock \bibinfo{publisher}{MIT Press}.
\newblock
\newblock
\shownote{\url{http://www.deeplearningbook.org}}.


\bibitem[Goodfellow et~al\mbox{.}(2015)]%
        {GoodfellowSS15}
\bibfield{author}{\bibinfo{person}{Ian~J. Goodfellow},
  \bibinfo{person}{Jonathon Shlens}, {and} \bibinfo{person}{Christian
  Szegedy}.} \bibinfo{year}{2015}\natexlab{}.
\newblock \showarticletitle{Explaining and Harnessing Adversarial Examples}. In
  \bibinfo{booktitle}{\emph{3rd International Conference on Learning
  Representations, {ICLR} 2015, San Diego, CA, USA, May 7-9, 2015, Conference
  Track Proceedings}}, \bibfield{editor}{\bibinfo{person}{Yoshua Bengio} {and}
  \bibinfo{person}{Yann LeCun}} (Eds.).
\newblock
\urldef\tempurl%
\url{http://arxiv.org/abs/1412.6572}
\showURL{%
\tempurl}


\bibitem[Grigorescu et~al\mbox{.}(2020)]%
        {GrigorescuTCM20}
\bibfield{author}{\bibinfo{person}{Sorin Grigorescu}, \bibinfo{person}{Bogdan
  Trasnea}, \bibinfo{person}{Tiberiu Cocias}, {and} \bibinfo{person}{Gigel
  Macesanu}.} \bibinfo{year}{2020}\natexlab{}.
\newblock \showarticletitle{A survey of deep learning techniques for autonomous
  driving}.
\newblock \bibinfo{journal}{\emph{Journal of Field Robotics}}
  \bibinfo{volume}{37}, \bibinfo{number}{3} (\bibinfo{year}{2020}),
  \bibinfo{pages}{362--386}.
\newblock
\urldef\tempurl%
\url{https://doi.org/10.1002/rob.21918}
\showDOI{\tempurl}
\showeprint{https://onlinelibrary.wiley.com/doi/pdf/10.1002/rob.21918}


\bibitem[Han et~al\mbox{.}(2016)]%
        {Han16}
\bibfield{author}{\bibinfo{person}{Song Han}, \bibinfo{person}{Huizi Mao},
  {and} \bibinfo{person}{William~J. Dally}.} \bibinfo{year}{2016}\natexlab{}.
\newblock \showarticletitle{Deep Compression: Compressing Deep Neural Network
  with Pruning, Trained Quantization and Huffman Coding}. In
  \bibinfo{booktitle}{\emph{4th International Conference on Learning
  Representations, {ICLR} 2016, San Juan, Puerto Rico, May 2-4, 2016,
  Conference Track Proceedings}}, \bibfield{editor}{\bibinfo{person}{Yoshua
  Bengio} {and} \bibinfo{person}{Yann LeCun}} (Eds.).
\newblock
\urldef\tempurl%
\url{http://arxiv.org/abs/1510.00149}
\showURL{%
\tempurl}


\bibitem[Hanspal and Lomuscio(2023)]%
        {HanspalL23}
\bibfield{author}{\bibinfo{person}{Harleen Hanspal} {and}
  \bibinfo{person}{Alessio Lomuscio}.} \bibinfo{year}{2023}\natexlab{}.
\newblock \showarticletitle{Efficient Verification of Neural Networks Against
  LVM-Based Specifications}. In \bibinfo{booktitle}{\emph{{IEEE/CVF} Conference
  on Computer Vision and Pattern Recognition, {CVPR} 2023, Vancouver, BC,
  Canada, June 17-24, 2023}}. \bibinfo{publisher}{{IEEE}},
  \bibinfo{pages}{3894--3903}.
\newblock
\urldef\tempurl%
\url{https://doi.org/10.1109/CVPR52729.2023.00379}
\showDOI{\tempurl}


\bibitem[He et~al\mbox{.}(2016)]%
        {He16}
\bibfield{author}{\bibinfo{person}{Kaiming He}, \bibinfo{person}{Xiangyu
  Zhang}, \bibinfo{person}{Shaoqing Ren}, {and} \bibinfo{person}{Jian Sun}.}
  \bibinfo{year}{2016}\natexlab{}.
\newblock \showarticletitle{Deep Residual Learning for Image Recognition}. In
  \bibinfo{booktitle}{\emph{2016 {IEEE} Conference on Computer Vision and
  Pattern Recognition, {CVPR} 2016}}. \bibinfo{publisher}{{IEEE} Computer
  Society}, \bibinfo{pages}{770--778}.
\newblock


\bibitem[He and Xiao(2024)]%
        {He24}
\bibfield{author}{\bibinfo{person}{Yang He} {and} \bibinfo{person}{Lingao
  Xiao}.} \bibinfo{year}{2024}\natexlab{}.
\newblock \showarticletitle{Structured Pruning for Deep Convolutional Neural
  Networks: {A} Survey}.
\newblock \bibinfo{journal}{\emph{{IEEE} Trans. Pattern Anal. Mach. Intell.}}
  \bibinfo{volume}{46}, \bibinfo{number}{5} (\bibinfo{year}{2024}),
  \bibinfo{pages}{2900--2919}.
\newblock


\bibitem[Huang et~al\mbox{.}(2020)]%
        {HuangKRSSTWY20}
\bibfield{author}{\bibinfo{person}{Xiaowei Huang}, \bibinfo{person}{Daniel
  Kroening}, \bibinfo{person}{Wenjie Ruan}, \bibinfo{person}{James Sharp},
  \bibinfo{person}{Youcheng Sun}, \bibinfo{person}{Emese Thamo},
  \bibinfo{person}{Min Wu}, {and} \bibinfo{person}{Xinping Yi}.}
  \bibinfo{year}{2020}\natexlab{}.
\newblock \showarticletitle{A survey of safety and trustworthiness of deep
  neural networks: Verification, testing, adversarial attack and defence, and
  interpretability}.
\newblock \bibinfo{journal}{\emph{Comput. Sci. Rev.}}  \bibinfo{volume}{37}
  (\bibinfo{year}{2020}), \bibinfo{pages}{100270}.
\newblock
\urldef\tempurl%
\url{https://doi.org/10.1016/J.COSREV.2020.100270}
\showDOI{\tempurl}


\bibitem[Imtiaz et~al\mbox{.}(2023)]%
        {ImtiazBSPCR23}
\bibfield{author}{\bibinfo{person}{Sayem~Mohammad Imtiaz},
  \bibinfo{person}{Fraol Batole}, \bibinfo{person}{Astha Singh},
  \bibinfo{person}{Rangeet Pan}, \bibinfo{person}{Breno~Dantas Cruz}, {and}
  \bibinfo{person}{Hridesh Rajan}.} \bibinfo{year}{2023}\natexlab{}.
\newblock \showarticletitle{Decomposing a Recurrent Neural Network into Modules
  for Enabling Reusability and Replacement}. In \bibinfo{booktitle}{\emph{45th
  {IEEE/ACM} International Conference on Software Engineering, {ICSE} 2023,
  Melbourne, Australia, May 14-20, 2023}}. \bibinfo{publisher}{{IEEE}},
  \bibinfo{pages}{1020--1032}.
\newblock


\bibitem[Jacobs et~al\mbox{.}(1991)]%
        {JacobsJNH91}
\bibfield{author}{\bibinfo{person}{Robert~A. Jacobs},
  \bibinfo{person}{Michael~I. Jordan}, \bibinfo{person}{Steven~J. Nowlan},
  {and} \bibinfo{person}{Geoffrey~E. Hinton}.} \bibinfo{year}{1991}\natexlab{}.
\newblock \showarticletitle{Adaptive Mixtures of Local Experts}.
\newblock \bibinfo{journal}{\emph{Neural Computation}} \bibinfo{volume}{3},
  \bibinfo{number}{1} (\bibinfo{date}{03} \bibinfo{year}{1991}),
  \bibinfo{pages}{79--87}.
\newblock
\showISSN{0899-7667}
\urldef\tempurl%
\url{https://doi.org/10.1162/neco.1991.3.1.79}
\showDOI{\tempurl}
\showeprint{https://direct.mit.edu/neco/article-pdf/3/1/79/812104/neco.1991.3.1.79.pdf}


\bibitem[Karimi and Derr(2022)]%
        {KarimiD22}
\bibfield{author}{\bibinfo{person}{Hamid Karimi} {and} \bibinfo{person}{Tyler
  Derr}.} \bibinfo{year}{2022}\natexlab{}.
\newblock \showarticletitle{Decision Boundaries of Deep Neural Networks}. In
  \bibinfo{booktitle}{\emph{2022 21st IEEE International Conference on Machine
  Learning and Applications (ICMLA)}}. \bibinfo{pages}{1085--1092}.
\newblock
\urldef\tempurl%
\url{https://doi.org/10.1109/ICMLA55696.2022.00179}
\showDOI{\tempurl}


\bibitem[Karimi et~al\mbox{.}(2019)]%
        {KarimiTD19}
\bibfield{author}{\bibinfo{person}{Hamid Karimi}, \bibinfo{person}{Tyler Derr},
  {and} \bibinfo{person}{Jiliang Tang}.} \bibinfo{year}{2019}\natexlab{}.
\newblock \showarticletitle{Characterizing the Decision Boundary of Deep Neural
  Networks}.
\newblock \bibinfo{journal}{\emph{CoRR}}  \bibinfo{volume}{abs/1912.11460}
  (\bibinfo{year}{2019}).
\newblock
\showeprint[arXiv]{1912.11460}
\urldef\tempurl%
\url{http://arxiv.org/abs/1912.11460}
\showURL{%
\tempurl}


\bibitem[Katz et~al\mbox{.}(2019)]%
        {KatzHIJLLSTWZDK19}
\bibfield{author}{\bibinfo{person}{Guy Katz}, \bibinfo{person}{Derek~A. Huang},
  \bibinfo{person}{Duligur Ibeling}, \bibinfo{person}{Kyle Julian},
  \bibinfo{person}{Christopher Lazarus}, \bibinfo{person}{Rachel Lim},
  \bibinfo{person}{Parth Shah}, \bibinfo{person}{Shantanu Thakoor},
  \bibinfo{person}{Haoze Wu}, \bibinfo{person}{Aleksandar Zeljic},
  \bibinfo{person}{David~L. Dill}, \bibinfo{person}{Mykel~J. Kochenderfer},
  {and} \bibinfo{person}{Clark~W. Barrett}.} \bibinfo{year}{2019}\natexlab{}.
\newblock \showarticletitle{The Marabou Framework for Verification and Analysis
  of Deep Neural Networks}. In \bibinfo{booktitle}{\emph{Computer Aided
  Verification - 31st International Conference, {CAV} 2019, New York City, NY,
  USA, July 15-18, 2019, Proceedings, Part {I}}}
  \emph{(\bibinfo{series}{Lecture Notes in Computer Science},
  Vol.~\bibinfo{volume}{11561})}, \bibfield{editor}{\bibinfo{person}{Isil
  Dillig} {and} \bibinfo{person}{Serdar Tasiran}} (Eds.).
  \bibinfo{publisher}{Springer}, \bibinfo{pages}{443--452}.
\newblock
\urldef\tempurl%
\url{https://doi.org/10.1007/978-3-030-25540-4\_26}
\showDOI{\tempurl}


\bibitem[Krizhevsky et~al\mbox{.}(2009)]%
        {krizhevsky2009learning}
\bibfield{author}{\bibinfo{person}{Alex Krizhevsky}, \bibinfo{person}{Geoffrey
  Hinton}, {et~al\mbox{.}}} \bibinfo{year}{2009}\natexlab{}.
\newblock \showarticletitle{Learning multiple layers of features from tiny
  images}.
\newblock  (\bibinfo{year}{2009}).
\newblock


\bibitem[LeCun et~al\mbox{.}(2015)]%
        {LeCunBH15}
\bibfield{author}{\bibinfo{person}{Yann LeCun}, \bibinfo{person}{Yoshua
  Bengio}, {and} \bibinfo{person}{Geoffrey Hinton}.}
  \bibinfo{year}{2015}\natexlab{}.
\newblock \showarticletitle{Deep learning}.
\newblock \bibinfo{journal}{\emph{Nature}} \bibinfo{volume}{521},
  \bibinfo{number}{7553} (\bibinfo{year}{2015}), \bibinfo{pages}{436--444}.
\newblock
\urldef\tempurl%
\url{https://doi.org/10.1038/nature14539}
\showDOI{\tempurl}


\bibitem[Li et~al\mbox{.}(2017)]%
        {Li17}
\bibfield{author}{\bibinfo{person}{Hao Li}, \bibinfo{person}{Asim Kadav},
  \bibinfo{person}{Igor Durdanovic}, \bibinfo{person}{Hanan Samet}, {and}
  \bibinfo{person}{Hans~Peter Graf}.} \bibinfo{year}{2017}\natexlab{}.
\newblock \showarticletitle{Pruning Filters for Efficient ConvNets}. In
  \bibinfo{booktitle}{\emph{5th International Conference on Learning
  Representations, {ICLR} 2017}}. \bibinfo{publisher}{OpenReview.net}.
\newblock


\bibitem[Madry et~al\mbox{.}(2018)]%
        {MadryMSTV18}
\bibfield{author}{\bibinfo{person}{Aleksander Madry},
  \bibinfo{person}{Aleksandar Makelov}, \bibinfo{person}{Ludwig Schmidt},
  \bibinfo{person}{Dimitris Tsipras}, {and} \bibinfo{person}{Adrian Vladu}.}
  \bibinfo{year}{2018}\natexlab{}.
\newblock \showarticletitle{Towards Deep Learning Models Resistant to
  Adversarial Attacks}. In \bibinfo{booktitle}{\emph{6th International
  Conference on Learning Representations, {ICLR} 2018, Vancouver, BC, Canada,
  April 30 - May 3, 2018, Conference Track Proceedings}}.
\newblock


\bibitem[Moosavi{-}Dezfooli et~al\mbox{.}(2016)]%
        {Moosavi-Dezfooli16}
\bibfield{author}{\bibinfo{person}{Seyed{-}Mohsen Moosavi{-}Dezfooli},
  \bibinfo{person}{Alhussein Fawzi}, {and} \bibinfo{person}{Pascal Frossard}.}
  \bibinfo{year}{2016}\natexlab{}.
\newblock \showarticletitle{DeepFool: {A} Simple and Accurate Method to Fool
  Deep Neural Networks}. In \bibinfo{booktitle}{\emph{2016 {IEEE} Conference on
  Computer Vision and Pattern Recognition, {CVPR} 2016, Las Vegas, NV, USA,
  June 27-30, 2016}}. \bibinfo{publisher}{{IEEE} Computer Society},
  \bibinfo{pages}{2574--2582}.
\newblock
\urldef\tempurl%
\url{https://doi.org/10.1109/CVPR.2016.282}
\showURL{%
\tempurl}


\bibitem[Odena et~al\mbox{.}(2019)]%
        {OdenaOAG19}
\bibfield{author}{\bibinfo{person}{Augustus Odena}, \bibinfo{person}{Catherine
  Olsson}, \bibinfo{person}{David~G. Andersen}, {and} \bibinfo{person}{Ian~J.
  Goodfellow}.} \bibinfo{year}{2019}\natexlab{}.
\newblock \showarticletitle{TensorFuzz: Debugging Neural Networks with
  Coverage-Guided Fuzzing}. In \bibinfo{booktitle}{\emph{Proceedings of the
  36th International Conference on Machine Learning, {ICML} 2019, 9-15 June
  2019, Long Beach, California, {USA}}} \emph{(\bibinfo{series}{Proceedings of
  Machine Learning Research}, Vol.~\bibinfo{volume}{97})},
  \bibfield{editor}{\bibinfo{person}{Kamalika Chaudhuri} {and}
  \bibinfo{person}{Ruslan Salakhutdinov}} (Eds.). \bibinfo{publisher}{{PMLR}},
  \bibinfo{pages}{4901--4911}.
\newblock
\urldef\tempurl%
\url{http://proceedings.mlr.press/v97/odena19a.html}
\showURL{%
\tempurl}


\bibitem[Otter et~al\mbox{.}(2020)]%
        {otter20}
\bibfield{author}{\bibinfo{person}{Daniel~W Otter}, \bibinfo{person}{Julian~R
  Medina}, {and} \bibinfo{person}{Jugal~K Kalita}.}
  \bibinfo{year}{2020}\natexlab{}.
\newblock \showarticletitle{A survey of the usages of deep learning for natural
  language processing}.
\newblock \bibinfo{journal}{\emph{IEEE transactions on neural networks and
  learning systems}} \bibinfo{volume}{32}, \bibinfo{number}{2}
  (\bibinfo{year}{2020}), \bibinfo{pages}{604--624}.
\newblock


\bibitem[Oyallon(2017)]%
        {Oyallon17}
\bibfield{author}{\bibinfo{person}{Edouard Oyallon}.}
  \bibinfo{year}{2017}\natexlab{}.
\newblock \showarticletitle{Building a Regular Decision Boundary with Deep
  Networks}. In \bibinfo{booktitle}{\emph{2017 {IEEE} Conference on Computer
  Vision and Pattern Recognition, {CVPR} 2017, Honolulu, HI, USA, July 21-26,
  2017}}. \bibinfo{publisher}{{IEEE} Computer Society},
  \bibinfo{pages}{1886--1894}.
\newblock
\urldef\tempurl%
\url{https://doi.org/10.1109/CVPR.2017.204}
\showDOI{\tempurl}


\bibitem[Pan and Rajan(2020)]%
        {PanR20}
\bibfield{author}{\bibinfo{person}{Rangeet Pan} {and} \bibinfo{person}{Hridesh
  Rajan}.} \bibinfo{year}{2020}\natexlab{}.
\newblock \showarticletitle{On decomposing a deep neural network into modules}.
  In \bibinfo{booktitle}{\emph{{ESEC/FSE} '20: 28th {ACM} Joint European
  Software Engineering Conference and Symposium on the Foundations of Software
  Engineering, Virtual Event, USA, November 8-13, 2020}},
  \bibfield{editor}{\bibinfo{person}{Prem Devanbu}, \bibinfo{person}{Myra~B.
  Cohen}, {and} \bibinfo{person}{Thomas Zimmermann}} (Eds.).
  \bibinfo{publisher}{{ACM}}, \bibinfo{pages}{889--900}.
\newblock
\urldef\tempurl%
\url{https://doi.org/10.1145/3368089.3409668}
\showDOI{\tempurl}


\bibitem[Pan and Rajan(2022)]%
        {PanR22}
\bibfield{author}{\bibinfo{person}{Rangeet Pan} {and} \bibinfo{person}{Hridesh
  Rajan}.} \bibinfo{year}{2022}\natexlab{}.
\newblock \showarticletitle{Decomposing Convolutional Neural Networks into
  Reusable and Replaceable Modules}. In \bibinfo{booktitle}{\emph{44th
  {IEEE/ACM} 44th International Conference on Software Engineering, {ICSE}
  2022, Pittsburgh, PA, USA, May 25-27, 2022}}. \bibinfo{publisher}{{ACM}},
  \bibinfo{pages}{524--535}.
\newblock
\urldef\tempurl%
\url{https://doi.org/10.1145/3510003.3510051}
\showDOI{\tempurl}


\bibitem[Ren et~al\mbox{.}(2023)]%
        {Ren23}
\bibfield{author}{\bibinfo{person}{Xiaoning Ren}, \bibinfo{person}{Yun Lin},
  \bibinfo{person}{Yinxing Xue}, \bibinfo{person}{Ruofan Liu},
  \bibinfo{person}{Jun Sun}, \bibinfo{person}{Zhiyong Feng}, {and}
  \bibinfo{person}{Jin~Song Dong}.} \bibinfo{year}{2023}\natexlab{}.
\newblock \showarticletitle{DeepArc: Modularizing Neural Networks for the Model
  Maintenance}. In \bibinfo{booktitle}{\emph{45th {IEEE/ACM} International
  Conference on Software Engineering, {ICSE} 2023}}.
  \bibinfo{publisher}{{IEEE}}, \bibinfo{pages}{1008--1019}.
\newblock


\bibitem[Sabour et~al\mbox{.}(2017)]%
        {10.5555/3294996.3295142}
\bibfield{author}{\bibinfo{person}{Sara Sabour}, \bibinfo{person}{Nicholas
  Frosst}, {and} \bibinfo{person}{Geoffrey~E. Hinton}.}
  \bibinfo{year}{2017}\natexlab{}.
\newblock \showarticletitle{Dynamic routing between capsules}. In
  \bibinfo{booktitle}{\emph{Proceedings of the 31st International Conference on
  Neural Information Processing Systems}} (Long Beach, California, USA)
  \emph{(\bibinfo{series}{NIPS'17})}. \bibinfo{publisher}{Curran Associates
  Inc.}, \bibinfo{address}{Red Hook, NY, USA}, \bibinfo{pages}{3859–3869}.
\newblock
\showISBNx{9781510860964}


\bibitem[Salman et~al\mbox{.}(2019)]%
        {salman2019convex}
\bibfield{author}{\bibinfo{person}{Hadi Salman}, \bibinfo{person}{Greg Yang},
  \bibinfo{person}{Huan Zhang}, \bibinfo{person}{Cho-Jui Hsieh}, {and}
  \bibinfo{person}{Pengchuan Zhang}.} \bibinfo{year}{2019}\natexlab{}.
\newblock \showarticletitle{A Convex Relaxation Barrier to Tight Robustness
  Verification of Neural Networks}. In \bibinfo{booktitle}{\emph{Advances in
  Neural Information Processing Systems}},
  \bibfield{editor}{\bibinfo{person}{H.~Wallach},
  \bibinfo{person}{H.~Larochelle}, \bibinfo{person}{A.~Beygelzimer},
  \bibinfo{person}{F.~d\textquotesingle Alch\'{e}-Buc},
  \bibinfo{person}{E.~Fox}, {and} \bibinfo{person}{R.~Garnett}} (Eds.),
  Vol.~\bibinfo{volume}{32}. \bibinfo{publisher}{Curran Associates, Inc.}
\newblock
\urldef\tempurl%
\url{https://proceedings.neurips.cc/paper_files/paper/2019/file/246a3c5544feb054f3ea718f61adfa16-Paper.pdf}
\showURL{%
\tempurl}


\bibitem[Sehwag et~al\mbox{.}(2020)]%
        {SehwagWMJ20}
\bibfield{author}{\bibinfo{person}{Vikash Sehwag}, \bibinfo{person}{Shiqi
  Wang}, \bibinfo{person}{Prateek Mittal}, {and} \bibinfo{person}{Suman Jana}.}
  \bibinfo{year}{2020}\natexlab{}.
\newblock \showarticletitle{HYDRA: Pruning Adversarially Robust Neural
  Networks}. In \bibinfo{booktitle}{\emph{Advances in Neural Information
  Processing Systems}}, \bibfield{editor}{\bibinfo{person}{H.~Larochelle},
  \bibinfo{person}{M.~Ranzato}, \bibinfo{person}{R.~Hadsell},
  \bibinfo{person}{M.F. Balcan}, {and} \bibinfo{person}{H.~Lin}} (Eds.),
  Vol.~\bibinfo{volume}{33}. \bibinfo{publisher}{Curran Associates, Inc.},
  \bibinfo{pages}{19655--19666}.
\newblock
\urldef\tempurl%
\url{https://proceedings.neurips.cc/paper_files/paper/2020/file/e3a72c791a69f87b05ea7742e04430ed-Paper.pdf}
\showURL{%
\tempurl}


\bibitem[Shamshad et~al\mbox{.}(2023)]%
        {Shamshad23}
\bibfield{author}{\bibinfo{person}{Fahad Shamshad}, \bibinfo{person}{Salman~H.
  Khan}, \bibinfo{person}{Syed~Waqas Zamir}, \bibinfo{person}{Muhammad~Haris
  Khan}, \bibinfo{person}{Munawar Hayat}, \bibinfo{person}{Fahad~Shahbaz Khan},
  {and} \bibinfo{person}{Huazhu Fu}.} \bibinfo{year}{2023}\natexlab{}.
\newblock \showarticletitle{Transformers in medical imaging: {A} survey}.
\newblock \bibinfo{journal}{\emph{Medical Image Anal.}}  \bibinfo{volume}{88}
  (\bibinfo{year}{2023}), \bibinfo{pages}{102802}.
\newblock


\bibitem[Shazeer et~al\mbox{.}(2017)]%
        {shazeer2017}
\bibfield{author}{\bibinfo{person}{Noam Shazeer}, \bibinfo{person}{Azalia
  Mirhoseini}, \bibinfo{person}{Krzysztof Maziarz}, \bibinfo{person}{Andy
  Davis}, \bibinfo{person}{Quoc~V. Le}, \bibinfo{person}{Geoffrey~E. Hinton},
  {and} \bibinfo{person}{Jeff Dean}.} \bibinfo{year}{2017}\natexlab{}.
\newblock \showarticletitle{Outrageously Large Neural Networks: The
  Sparsely-Gated Mixture-of-Experts Layer}. In \bibinfo{booktitle}{\emph{5th
  International Conference on Learning Representations, {ICLR} 2017, Toulon,
  France, April 24-26, 2017, Conference Track Proceedings}}.
  \bibinfo{publisher}{OpenReview.net}.
\newblock
\urldef\tempurl%
\url{https://openreview.net/forum?id=B1ckMDqlg}
\showURL{%
\tempurl}


\bibitem[She et~al\mbox{.}(2019)]%
        {8835342}
\bibfield{author}{\bibinfo{person}{Dongdong She}, \bibinfo{person}{Kexin Pei},
  \bibinfo{person}{Dave Epstein}, \bibinfo{person}{Junfeng Yang},
  \bibinfo{person}{Baishakhi Ray}, {and} \bibinfo{person}{Suman Jana}.}
  \bibinfo{year}{2019}\natexlab{}.
\newblock \showarticletitle{NEUZZ: Efficient Fuzzing with Neural Program
  Smoothing}. In \bibinfo{booktitle}{\emph{2019 IEEE Symposium on Security and
  Privacy (SP)}}. \bibinfo{pages}{803--817}.
\newblock
\urldef\tempurl%
\url{https://doi.org/10.1109/SP.2019.00052}
\showDOI{\tempurl}


\bibitem[Singh et~al\mbox{.}(2018)]%
        {SinghGMPV18}
\bibfield{author}{\bibinfo{person}{Gagandeep Singh}, \bibinfo{person}{Timon
  Gehr}, \bibinfo{person}{Matthew Mirman}, \bibinfo{person}{Markus
  P{\"{u}}schel}, {and} \bibinfo{person}{Martin~T. Vechev}.}
  \bibinfo{year}{2018}\natexlab{}.
\newblock \showarticletitle{Fast and Effective Robustness Certification}. In
  \bibinfo{booktitle}{\emph{Advances in Neural Information Processing Systems
  31: Annual Conference on Neural Information Processing Systems 2018, NeurIPS
  2018, December 3-8, 2018, Montr{\'{e}}al, Canada}},
  \bibfield{editor}{\bibinfo{person}{Samy Bengio}, \bibinfo{person}{Hanna~M.
  Wallach}, \bibinfo{person}{Hugo Larochelle}, \bibinfo{person}{Kristen
  Grauman}, \bibinfo{person}{Nicol{\`{o}} Cesa{-}Bianchi}, {and}
  \bibinfo{person}{Roman Garnett}} (Eds.). \bibinfo{pages}{10825--10836}.
\newblock


\bibitem[Singh et~al\mbox{.}(2019)]%
        {SinghGPV19}
\bibfield{author}{\bibinfo{person}{Gagandeep Singh}, \bibinfo{person}{Timon
  Gehr}, \bibinfo{person}{Markus P{\"{u}}schel}, {and}
  \bibinfo{person}{Martin~T. Vechev}.} \bibinfo{year}{2019}\natexlab{}.
\newblock \showarticletitle{An abstract domain for certifying neural networks}.
\newblock \bibinfo{journal}{\emph{Proc. {ACM} Program. Lang.}}
  \bibinfo{volume}{3}, \bibinfo{number}{{POPL}} (\bibinfo{year}{2019}),
  \bibinfo{pages}{41:1--41:30}.
\newblock


\bibitem[Somepalli et~al\mbox{.}(2022)]%
        {SomepalliFBCDBG22}
\bibfield{author}{\bibinfo{person}{Gowthami Somepalli}, \bibinfo{person}{Liam
  Fowl}, \bibinfo{person}{Arpit Bansal}, \bibinfo{person}{Ping{-}Yeh Chiang},
  \bibinfo{person}{Yehuda Dar}, \bibinfo{person}{Richard~G. Baraniuk},
  \bibinfo{person}{Micah Goldblum}, {and} \bibinfo{person}{Tom Goldstein}.}
  \bibinfo{year}{2022}\natexlab{}.
\newblock \showarticletitle{Can Neural Nets Learn the Same Model Twice?
  Investigating Reproducibility and Double Descent from the Decision Boundary
  Perspective}. In \bibinfo{booktitle}{\emph{{IEEE/CVF} Conference on Computer
  Vision and Pattern Recognition, {CVPR} 2022, New Orleans, LA, USA, June
  18-24, 2022}}. \bibinfo{publisher}{{IEEE}}, \bibinfo{pages}{13689--13698}.
\newblock
\urldef\tempurl%
\url{https://doi.org/10.1109/CVPR52688.2022.01333}
\showDOI{\tempurl}


\bibitem[Sun et~al\mbox{.}(2024)]%
        {Sun0BK24}
\bibfield{author}{\bibinfo{person}{Mingjie Sun}, \bibinfo{person}{Zhuang Liu},
  \bibinfo{person}{Anna Bair}, {and} \bibinfo{person}{J.~Zico Kolter}.}
  \bibinfo{year}{2024}\natexlab{}.
\newblock \showarticletitle{A Simple and Effective Pruning Approach for Large
  Language Models}. In \bibinfo{booktitle}{\emph{The Twelfth International
  Conference on Learning Representations, {ICLR} 2024, Vienna, Austria, May
  7-11, 2024}}. \bibinfo{publisher}{OpenReview.net}.
\newblock
\urldef\tempurl%
\url{https://openreview.net/forum?id=PxoFut3dWW}
\showURL{%
\tempurl}


\bibitem[Touvron et~al\mbox{.}(2021)]%
        {pmlr-v139-touvron21a}
\bibfield{author}{\bibinfo{person}{Hugo Touvron}, \bibinfo{person}{Matthieu
  Cord}, \bibinfo{person}{Matthijs Douze}, \bibinfo{person}{Francisco Massa},
  \bibinfo{person}{Alexandre Sablayrolles}, {and} \bibinfo{person}{Herve
  Jegou}.} \bibinfo{year}{2021}\natexlab{}.
\newblock \showarticletitle{Training data-efficient image transformers \&amp;
  distillation through attention}. In \bibinfo{booktitle}{\emph{Proceedings of
  the 38th International Conference on Machine Learning}}
  \emph{(\bibinfo{series}{Proceedings of Machine Learning Research},
  Vol.~\bibinfo{volume}{139})}, \bibfield{editor}{\bibinfo{person}{Marina
  Meila} {and} \bibinfo{person}{Tong Zhang}} (Eds.). \bibinfo{publisher}{PMLR},
  \bibinfo{pages}{10347--10357}.
\newblock
\urldef\tempurl%
\url{https://proceedings.mlr.press/v139/touvron21a.html}
\showURL{%
\tempurl}


\bibitem[Ugare et~al\mbox{.}(2022)]%
        {Shubham22}
\bibfield{author}{\bibinfo{person}{Shubham Ugare}, \bibinfo{person}{Gagandeep
  Singh}, {and} \bibinfo{person}{Sasa Misailovic}.}
  \bibinfo{year}{2022}\natexlab{}.
\newblock \showarticletitle{Proof transfer for fast certification of multiple
  approximate neural networks}.
\newblock \bibinfo{journal}{\emph{Proc. ACM Program. Lang.}}
  \bibinfo{volume}{6}, \bibinfo{number}{OOPSLA1}, Article
  \bibinfo{articleno}{75} (\bibinfo{date}{apr} \bibinfo{year}{2022}),
  \bibinfo{numpages}{29}~pages.
\newblock
\urldef\tempurl%
\url{https://doi.org/10.1145/3527319}
\showDOI{\tempurl}


\bibitem[Vaswani et~al\mbox{.}(2017)]%
        {VaswaniSUAJPGP17}
\bibfield{author}{\bibinfo{person}{Ashish Vaswani}, \bibinfo{person}{Noam
  Shazeer}, \bibinfo{person}{Niki Parmar}, \bibinfo{person}{Jakob Uszkoreit},
  \bibinfo{person}{Llion Jones}, \bibinfo{person}{Aidan~N Gomez},
  \bibinfo{person}{\L~ukasz Kaiser}, {and} \bibinfo{person}{Illia Polosukhin}.}
  \bibinfo{year}{2017}\natexlab{}.
\newblock \showarticletitle{Attention is All you Need}. In
  \bibinfo{booktitle}{\emph{Advances in Neural Information Processing
  Systems}}, \bibfield{editor}{\bibinfo{person}{I.~Guyon},
  \bibinfo{person}{U.~Von Luxburg}, \bibinfo{person}{S.~Bengio},
  \bibinfo{person}{H.~Wallach}, \bibinfo{person}{R.~Fergus},
  \bibinfo{person}{S.~Vishwanathan}, {and} \bibinfo{person}{R.~Garnett}}
  (Eds.), Vol.~\bibinfo{volume}{30}. \bibinfo{publisher}{Curran Associates,
  Inc.}
\newblock
\urldef\tempurl%
\url{https://proceedings.neurips.cc/paper_files/paper/2017/file/3f5ee243547dee91fbd053c1c4a845aa-Paper.pdf}
\showURL{%
\tempurl}


\bibitem[Wang et~al\mbox{.}(2021)]%
        {wang2021beta}
\bibfield{author}{\bibinfo{person}{Shiqi Wang}, \bibinfo{person}{Huan Zhang},
  \bibinfo{person}{Kaidi Xu}, \bibinfo{person}{Xue Lin}, \bibinfo{person}{Suman
  Jana}, \bibinfo{person}{Cho-Jui Hsieh}, {and} \bibinfo{person}{J.~Zico
  Kolter}.} \bibinfo{year}{2021}\natexlab{}.
\newblock \showarticletitle{Beta-CROWN: Efficient Bound Propagation with
  Per-neuron Split Constraints for Neural Network Robustness Verification}. In
  \bibinfo{booktitle}{\emph{Advances in Neural Information Processing
  Systems}}, \bibfield{editor}{\bibinfo{person}{M.~Ranzato},
  \bibinfo{person}{A.~Beygelzimer}, \bibinfo{person}{Y.~Dauphin},
  \bibinfo{person}{P.S. Liang}, {and} \bibinfo{person}{J.~Wortman Vaughan}}
  (Eds.), Vol.~\bibinfo{volume}{34}. \bibinfo{publisher}{Curran Associates,
  Inc.}, \bibinfo{pages}{29909--29921}.
\newblock
\urldef\tempurl%
\url{https://proceedings.neurips.cc/paper_files/paper/2021/file/fac7fead96dafceaf80c1daffeae82a4-Paper.pdf}
\showURL{%
\tempurl}


\bibitem[Wu et~al\mbox{.}(2024)]%
        {WuIZTDKRAJBHLWZKKB24}
\bibfield{author}{\bibinfo{person}{Haoze Wu}, \bibinfo{person}{Omri Isac},
  \bibinfo{person}{Aleksandar Zeljic}, \bibinfo{person}{Teruhiro Tagomori},
  \bibinfo{person}{Matthew~L. Daggitt}, \bibinfo{person}{Wen Kokke},
  \bibinfo{person}{Idan Refaeli}, \bibinfo{person}{Guy Amir},
  \bibinfo{person}{Kyle Julian}, \bibinfo{person}{Shahaf Bassan},
  \bibinfo{person}{Pei Huang}, \bibinfo{person}{Ori Lahav},
  \bibinfo{person}{Min Wu}, \bibinfo{person}{Min Zhang},
  \bibinfo{person}{Ekaterina Komendantskaya}, \bibinfo{person}{Guy Katz}, {and}
  \bibinfo{person}{Clark~W. Barrett}.} \bibinfo{year}{2024}\natexlab{}.
\newblock \showarticletitle{Marabou 2.0: {A} Versatile Formal Analyzer of
  Neural Networks}. In \bibinfo{booktitle}{\emph{Computer Aided Verification -
  36th International Conference, {CAV} 2024, Montreal, QC, Canada, July 24-27,
  2024, Proceedings, Part {II}}} \emph{(\bibinfo{series}{Lecture Notes in
  Computer Science}, Vol.~\bibinfo{volume}{14682})},
  \bibfield{editor}{\bibinfo{person}{Arie Gurfinkel} {and}
  \bibinfo{person}{Vijay Ganesh}} (Eds.). \bibinfo{publisher}{Springer},
  \bibinfo{pages}{249--264}.
\newblock
\urldef\tempurl%
\url{https://doi.org/10.1007/978-3-031-65630-9\_13}
\showDOI{\tempurl}


\bibitem[Xu et~al\mbox{.}(2020)]%
        {xu2020automatic}
\bibfield{author}{\bibinfo{person}{Kaidi Xu}, \bibinfo{person}{Zhouxing Shi},
  \bibinfo{person}{Huan Zhang}, \bibinfo{person}{Yihan Wang},
  \bibinfo{person}{Kai-Wei Chang}, \bibinfo{person}{Minlie Huang},
  \bibinfo{person}{Bhavya Kailkhura}, \bibinfo{person}{Xue Lin}, {and}
  \bibinfo{person}{Cho-Jui Hsieh}.} \bibinfo{year}{2020}\natexlab{}.
\newblock \showarticletitle{Automatic Perturbation Analysis for Scalable
  Certified Robustness and Beyond}. In \bibinfo{booktitle}{\emph{Advances in
  Neural Information Processing Systems}},
  \bibfield{editor}{\bibinfo{person}{H.~Larochelle},
  \bibinfo{person}{M.~Ranzato}, \bibinfo{person}{R.~Hadsell},
  \bibinfo{person}{M.F. Balcan}, {and} \bibinfo{person}{H.~Lin}} (Eds.),
  Vol.~\bibinfo{volume}{33}. \bibinfo{publisher}{Curran Associates, Inc.},
  \bibinfo{pages}{1129--1141}.
\newblock
\urldef\tempurl%
\url{https://proceedings.neurips.cc/paper_files/paper/2020/file/0cbc5671ae26f67871cb914d81ef8fc1-Paper.pdf}
\showURL{%
\tempurl}


\bibitem[Xu et~al\mbox{.}(2021)]%
        {xu2021fast}
\bibfield{author}{\bibinfo{person}{Kaidi Xu}, \bibinfo{person}{Huan Zhang},
  \bibinfo{person}{Shiqi Wang}, \bibinfo{person}{Yihan Wang},
  \bibinfo{person}{Suman Jana}, \bibinfo{person}{Xue Lin}, {and}
  \bibinfo{person}{Cho-Jui Hsieh}.} \bibinfo{year}{2021}\natexlab{}.
\newblock \showarticletitle{{Fast and Complete}: Enabling Complete Neural
  Network Verification with Rapid and Massively Parallel Incomplete Verifiers}.
  In \bibinfo{booktitle}{\emph{International Conference on Learning
  Representations}}.
\newblock
\urldef\tempurl%
\url{https://openreview.net/forum?id=nVZtXBI6LNn}
\showURL{%
\tempurl}


\bibitem[Xu et~al\mbox{.}(2023)]%
        {DBLP:conf/iclr/XuSGGH23}
\bibfield{author}{\bibinfo{person}{Yuancheng Xu}, \bibinfo{person}{Yanchao
  Sun}, \bibinfo{person}{Micah Goldblum}, \bibinfo{person}{Tom Goldstein},
  {and} \bibinfo{person}{Furong Huang}.} \bibinfo{year}{2023}\natexlab{}.
\newblock \showarticletitle{Exploring and Exploiting Decision Boundary Dynamics
  for Adversarial Robustness}. In \bibinfo{booktitle}{\emph{The Eleventh
  International Conference on Learning Representations, {ICLR} 2023, Kigali,
  Rwanda, May 1-5, 2023}}. \bibinfo{publisher}{OpenReview.net}.
\newblock
\urldef\tempurl%
\url{https://openreview.net/forum?id=aRTKuscKByJ}
\showURL{%
\tempurl}


\bibitem[Zhang et~al\mbox{.}(2018)]%
        {zhang2018efficient}
\bibfield{author}{\bibinfo{person}{Huan Zhang}, \bibinfo{person}{Tsui-Wei
  Weng}, \bibinfo{person}{Pin-Yu Chen}, \bibinfo{person}{Cho-Jui Hsieh}, {and}
  \bibinfo{person}{Luca Daniel}.} \bibinfo{year}{2018}\natexlab{}.
\newblock \showarticletitle{Efficient Neural Network Robustness Certification
  with General Activation Functions}. In \bibinfo{booktitle}{\emph{Advances in
  Neural Information Processing Systems}},
  \bibfield{editor}{\bibinfo{person}{S.~Bengio}, \bibinfo{person}{H.~Wallach},
  \bibinfo{person}{H.~Larochelle}, \bibinfo{person}{K.~Grauman},
  \bibinfo{person}{N.~Cesa-Bianchi}, {and} \bibinfo{person}{R.~Garnett}}
  (Eds.), Vol.~\bibinfo{volume}{31}. \bibinfo{publisher}{Curran Associates,
  Inc.}
\newblock
\urldef\tempurl%
\url{https://proceedings.neurips.cc/paper_files/paper/2018/file/d04863f100d59b3eb688a11f95b0ae60-Paper.pdf}
\showURL{%
\tempurl}


\bibitem[Zhang et~al\mbox{.}(2015)]%
        {ZhangZL15}
\bibfield{author}{\bibinfo{person}{Xiang Zhang}, \bibinfo{person}{Junbo~Jake
  Zhao}, {and} \bibinfo{person}{Yann LeCun}.} \bibinfo{year}{2015}\natexlab{}.
\newblock \showarticletitle{Character-level Convolutional Networks for Text
  Classification}. In \bibinfo{booktitle}{\emph{Advances in Neural Information
  Processing Systems 28: Annual Conference on Neural Information Processing
  Systems 2015, December 7-12, 2015, Montreal, Quebec, Canada}},
  \bibfield{editor}{\bibinfo{person}{Corinna Cortes}, \bibinfo{person}{Neil~D.
  Lawrence}, \bibinfo{person}{Daniel~D. Lee}, \bibinfo{person}{Masashi
  Sugiyama}, {and} \bibinfo{person}{Roman Garnett}} (Eds.).
  \bibinfo{pages}{649--657}.
\newblock
\urldef\tempurl%
\url{https://proceedings.neurips.cc/paper/2015/hash/250cf8b51c773f3f8dc8b4be867a9a02-Abstract.html}
\showURL{%
\tempurl}


\bibitem[Zhao et~al\mbox{.}(2024)]%
        {ZhaoNG24}
\bibfield{author}{\bibinfo{person}{Siyan Zhao}, \bibinfo{person}{Tung Nguyen},
  {and} \bibinfo{person}{Aditya Grover}.} \bibinfo{year}{2024}\natexlab{}.
\newblock \showarticletitle{Probing the Decision Boundaries of In-context
  Learning in Large Language Models}. In \bibinfo{booktitle}{\emph{Advances in
  Neural Information Processing Systems}},
  \bibfield{editor}{\bibinfo{person}{A.~Globerson},
  \bibinfo{person}{L.~Mackey}, \bibinfo{person}{D.~Belgrave},
  \bibinfo{person}{A.~Fan}, \bibinfo{person}{U.~Paquet},
  \bibinfo{person}{J.~Tomczak}, {and} \bibinfo{person}{C.~Zhang}} (Eds.),
  Vol.~\bibinfo{volume}{37}. \bibinfo{publisher}{Curran Associates, Inc.},
  \bibinfo{pages}{130408--130432}.
\newblock
\urldef\tempurl%
\url{https://doi.org/10.52202/079017-4144}
\showDOI{\tempurl}


\end{thebibliography}

\end{document}